\documentclass[a4paper,UKenglish, numberwithinsect, cleveref, autoref, thm-restate]{lipics-v2021}

\pdfoutput=1 
\hideLIPIcs  

\title{A Simple Obligation to Metric Interval Temporal Logic} 
\date{}

\usepackage{amsmath, amssymb}
\usepackage{algorithm}
\usepackage{algpseudocode}
\usepackage{mathtools}
\usepackage{thmtools}
\usepackage{thm-restate}
\usepackage{stmaryrd}
\usepackage{tikz}
\usetikzlibrary{automata, arrows.meta, positioning, shapes.geometric}
\usepackage{caption}
\usepackage{enumerate}
\usepackage{todonotes}

\newcommand{\incl}{\subseteq}
\newcommand{\Rpos}{\mathbb{R}_{\geq 0}}
\newcommand{\xra}[1]{\xrightarrow{#1}}

\newcommand{\N}{\mathbb{N}}

\newcommand{\U}{\mathcal{U}}
\newcommand{\weakuntil}{\mathcal{WU}}
\newcommand{\X}{\mathbf{X}}
\newcommand{\G}{\mathbf{G}}

\newcommand{\obred}{\mathsf{Obl}^{\mathsf{red}}}
\newcommand{\Thetaopt}{\Theta_{\mathsf{opt}}}
\newcommand{\idopt}{\id_{\mathsf{opt}}}

\newcommand{\Ii}{\mathcal{O}}
\newcommand{\Oo}{\mathcal{O}}

\newcommand{\F}{\mathbf{F}}

\newcommand{\range}{\mathsf{range}}

\newcommand{\init}{\operatorname{init}}
\newcommand{\ob}{\mathsf{Obl}}
\newcommand{\id}{\mathsf{id}}
\newcommand{\identifiers}{\mathsf{Identifiers}}
\newcommand{\Acc}{\mathsf{Acc}}

\newcommand{\optimize}{\mathsf{Optimize}}

\newcommand{\type}{\mathsf{type}}
\newcommand{\modifyold}{\textsf{Modify-old}}
\newcommand{\removenew}{\textsf{Remove-new}}
\newcommand{\mmerge}{\textsf{Merge-overlapping}}
\newcommand{\keepboth}{\textsf{Keep-both}}

\newcommand{\regeq}{\sim}
\newcommand{\integraltimer}[1]{\lceil #1 \rceil}
\newcommand{\integralclock}[1]{\lfloor #1 \rfloor}
\newcommand{\fractional}[1]{\{ #1 \}}

\let\regeq\simeq
\let\sat\models
\let\dual\widetilde

\author{Patricia Bouyer}{Universit\'e Paris-Saclay, CNRS, ENS Paris-Saclay, LMF, Gif-sur-Yvette,
France}{bouyer@lmf.cnrs.fr}{}{}
\author{B Srivathsan}{Chennai Mathematical Institute, India \and CNRS, ReLaX, IRL 2000, Siruseri, India}{sri@cmi.ac.in}{}{}

\author{Vaishnavi Vishwanath}{Chennai Mathematical Institute, India}{vaishnaviv@cmi.ac.in}{}{}

\authorrunning{Patricia Bouyer, B Srivathsan,
Vaishnavi Vishwanath} 
\Copyright{Patricia Bouyer, B Srivathsan,
Vaishnavi Vishwanath}


\nolinenumbers

\begin{document}

\maketitle
\begin{abstract}
Satisfiability of Metric Interval Temporal Logic (MITL) is a widely investigated subject. In this work, we present a new, and arguably simpler, approach for MITL satisfiability, based on an idea of tracking time-constrained obligations along a word.

To check whether a Linear Temporal Logic (LTL) formula is true at a position of a word, it is natural to generate certain obligations that need to be satisfied at a later point. For instance, $a~\U~b$ (with strict Until semantics) is true at position $i$ if either $b$ or the set $\{a, a~\U~b\}$ is true at $i+1$. We enhance this idea in the context of MITL by introducing a notion of time inside these obligations. However, a naïve procedure could lead to more and more obligations getting generated along the word, with no bound on the number. We propose a simple mechanism to eliminate or merge redundant obligations. For MITL, this mechanism ensures that only a bounded number of obligations are maintained along the entire timed word. We develop this observation into a symbolic procedure for MITL satisfiability using regions.
\end{abstract}

\section{Introduction}
\label{sec:intro}

Metric Temporal Logic~\cite{koymans1990specifying} is an extension of Linear Temporal Logic (LTL) to reason about timing constraints between events. Syntactically, it is obtained by adding intervals to the temporal operators, for instance the Until operator $\U$ of LTL becomes $\U_I$ in MTL, with $I$ being an interval over the reals that is bounded by natural numbers or $\infty$. MTL can be interpreted over timed words (pointwise semantics), or over timed signals (continuous semantics). For comparisons between these semantics we refer the reader to~\cite{DBLP:journals/sttt/DSouzaP07,DBLP:journals/corr/abs-2603-15379}. \emph{In this work, we will consider only the pointwise semantics.}

Over the years, it has been established that MTL satisfiablity is undecidable over infinite words~\cite{OW06}, whereas it is decidable over finite words, albeit with non-primitive recursive complexity~\cite{DBLP:conf/lics/OuaknineW05,ouaknine2007decidability}. A restriction of MTL, called \emph{Metric Interval Temporal Logic (MITL)} limits the syntax by disallowing singleton intervals, i.e., intervals of the form $[l,l]$. In a fundamental work~\cite{alur1996benefits} it was shown that MITL satisfiability is EXPSPACE-complete for both finite and infinite words. This positive result has led to several improvements and the development of tools for MITL satisfiability and model-checking~\cite{brihaye2013mitl,DBLP:conf/formats/BrihayeEG14,DBLP:conf/cav/BrihayeGHM17,DBLP:conf/tacas/HoKMMP26,DBLP:conf/concur/0001G0S24,DBLP:conf/tacas/AkshayCGGS26}.

Satisfiability procedures typically involve a translation of MITL to some form of automata: the very first work~\cite{alur1996benefits} translates MITL formulas to equivalent timed automata~\cite{alur1994theory}; subsequent work \cite{brihaye2013mitl} and its extensions translate MITL to a model of 1-clock alternating timed automata, which are then compiled to a network of timed automata; the recent work \cite{DBLP:conf/concur/0001G0S24} translates MITL to generalized timed automata~\cite{DBLP:conf/cav/AkshayGGJS23}. A common perception in this line of work is that the translations can be difficult to understand, as they require a careful case analysis and/or rely on sophisticated machinery. This is particularly the case for the Until operator $\U_I$ when $I$ is a two-sided interval, with none of the end-points being $0$ or $\infty$.  

In this paper, we present a new technique for MITL satisfiability, which incidentally does not go via an automaton translation. We believe our technique is simpler to understand as it does not require any intricate transformation 
or heavy machinery. 
Our approach stays close to mimicking the semantics of satisfaction of formulas. To assert whether a formula is true at a point in the word, one can generate an \emph{obligation} which needs to be settled along the word. For instance, the LTL formula $a~\U~b$ (with strict Until semantics) is true at position $i$ if either $b$ or the set $\{a, a~\U~b\}$ is true at $i+1$: put in other words, the obligation $a~\U~b$ at position $i$ generates either an obligation $b$ or a set of obligations $\{a, a~\U~b\}$ at position $i+1$. We enrich this notion of obligations by storing some kind of timing information and build a transition system whose accepting runs precisely capture the set of timed words satisfying the formula (Section~\ref{sec:obligations}). Furthermore, for MITL, we prove that it is sufficient to maintain a bounded set of obligations at each point in the word: the key contrast w.r.t. the existing approaches is the fact that we arrive to this bound by simple rules to merge or eliminate redundant obligations (Section~\ref{sec:removing-redundant-obligations}). Using this bound, we develop a finite region-abstraction over these obligations which leads to a symbolic algorithm for MITL satisfiability (Section~\ref{sec:symbolic}). After recalling preliminaries in Section~\ref{sec:mitl-prelims}, we present in Section~\ref{sec:overview} a detailed overview of the existing approaches followed by a brief introduction to our approach.

\section{Preliminaries}
\label{sec:mitl-prelims}
Let $\Sigma$ be a finite alphabet. Let $\N$ and $\Rpos$ denote the set of natural numbers (including $0$) and non-negative reals, respectively. A \emph{timed word} is an infinite sequence $w = (a_1,\tau_1)(a_2,\tau_2)\dots$ where $a_i \in \Sigma$ for each $i \in \N$, and $\tau_i \in \Rpos$ with $\tau_i \le \tau_{i+1}$ for all $i \in \N$. The number $\tau_i$ denotes the timestamp of occurrence of the action $a_i$.  A timed word is \emph{non-Zeno} if the sequence $\{\tau_i\}_{i \ge 1}$ of timestamps is diverging. Therefore, in a non-Zeno timed word, every unit interval contains only finitely many actions.

\smallskip 
\noindent \textbf{\emph{Metric Temporal Logic (MTL).}} The syntax of MTL is described using the grammar:  
\begin{align*}
  \varphi:= \top \mid \bot \mid a \mid \lnot a \mid \varphi \land \varphi \mid \varphi \lor \varphi \mid \lnot \varphi \mid \varphi ~\U_{I}~ \varphi \mid \varphi ~\dual{\U}_{I}~\varphi
\end{align*} 
where $a \in \Sigma$, $I \incl \Rpos$ is a closed, open or half-closed interval with endpoints in $\N \cup \{\infty\}$; $\U$ is called the \emph{Until} operator, and $\dual{\U}$ is called the \emph{dual Until} operator. 
Given a timed word $w = (a_1,\tau_1)(a_2,\tau_2)\dots$, a position $i \ge 1$ and an MTL formula $\varphi$, we define when $w$ satisfies $\varphi$ at position $i$, written as $(w,i) \sat \varphi$, as follows:
\begin{itemize}
    \item $(w, i) \sat \top$, $(w, i) \not \sat \bot$
\item $(w, i) \sat a$ if $a_{i} = a$, and $(w,i) \sat \lnot a$ if $a_{i} \neq a$
\item $(w,i) \sat \varphi_{1} \land \varphi_{2}$ if $(w,i) \sat \varphi_{1}$ and $(w,i) \sat \varphi_{2}$
\item $(w,i) \sat \varphi_{1} \lor \varphi_{2}$ if $(w,i) \sat \varphi_{1}$ or $(w,i) \sat \varphi_{2}$ 
\item $(w,i) \sat \lnot \varphi$ if $(w,i) \not\sat \varphi$
\item $(w,i) \sat \varphi_1 ~\U_{I}~ \varphi_{2}$ if there is some $j > i$ such that $(w,j) \sat \varphi_{2}$, $\tau_j - \tau_i \in I$ and for all $i < k < j$, we have $(w,k) \sat \varphi_1$
\item $(w,i) \sat \varphi_1 ~\dual{\U}_I~ \varphi_2$ if for all $j$ s.t. $j > i$ and $\tau_j - \tau_i \in I$, either $(w,j) \sat \varphi_{2}$, or there is some $k$ with $i < k < j$ with $(w,k) \sat \varphi_{1}$
\end{itemize}
Observe that we make use of a strict semantics for the $\U_I$ operator as in~\cite{DBLP:conf/lics/OuaknineW05,DBLP:conf/concur/Henzinger98,DBLP:conf/tacas/HoKMMP26}: we say $j > i$ (strictly greater) and $i < k < j$. The weak semantics (with $j \ge i$ and $i \le k < j$) can be recovered using the strict Until operator. For example, writing $\weakuntil$ for the until operator with a non-strict semantics, we see that $(a ~\weakuntil_{[0,5]}~ b) = b \lor  (a \land (a~\U_{[0, 5]}~b))$. 

Notice from the semantics that: $\neg (\varphi_1~\U_I~\varphi_2) = (\neg \varphi_1)~\dual{\U}_I~(\neg \varphi_2)$, and $\neg (\varphi_1~\dual{\U}_I~\varphi_2) = (\neg \varphi_1)~\U_I~(\neg \varphi_2)$. Hence any formula of the form $\neg \varphi$ can be converted into a negation normal form where all the negations appear at the atomic level and whose size is linear in the size of the original formula. From now on, we assume formulas are in negation normal form wlog. 
Here are some derived operators for convenience~\cite{DBLP:conf/lics/OuaknineW05}:  constrained next time operator $\X_I\varphi \equiv \bot~\U_I~\varphi$, constrained eventually $\F_I \varphi \equiv \top \U_I \varphi$, and constrained always $\G_I \varphi \equiv \neg (\top ~\U_I~(\neg\varphi)) \equiv \bot~\dual{\U}_I~\varphi$. Furthermore, if $I$ is $\Rpos$, we might omit the subscript.

The \emph{satisfiability} question for MTL is as follows: given an MTL formula $\varphi$ does there exist a non-Zeno timed word $w$ such that $w \sat \varphi$. This problem is known to be undecidable~\cite{OW06}.

\smallskip

\noindent \textbf{\emph{Metric Interval Temporal Logic (MITL).}} An interval $I$ is said to be \emph{singular} if it contains a single point, i.e., it is of the form $[l,l]$ for some $l \in \N$. MITL is a restriction of MTL where all intervals used are \emph{nonsingular}: for e.g., $a ~\U_{[2,2]} b$ is not an MITL formula. 
For MITL, the satisfiablity problem was shown to be EXPSPACE-complete in~\cite{alur1996benefits}. Their algorithm converts an MITL formula into an equivalent timed automaton, equipped with a generalized B\"uchi condition.  
Although the procedure is described for the continuous semantics (that uses timed state sequences instead of timed words), it is usually believed that the result also holds for the pointwise semantics that uses timed words as we do in this paper (for instance, mentioned in~\cite[page 119]{alur1996benefits}). Either way, our symbolic abstraction in Section~\ref{sec:symbolic} explicitly proves the EXPSPACE upper bound.

\section{Existing constructions for MITL and overview of our approach}
\label{sec:overview}

Let us start with the untimed LTL formula of the form $\varphi_1 \U \varphi_2$, and an untimed word $a_1 a_2 \cdots$. To assert $\varphi_1 \U \varphi_2$ at position $i$, the procedure needs to check if there exists a $j > i$ that satisfies $\varphi_2$ and $\varphi_1$ holds at all positions $k$ such that $i < k < j$ (recall we are using strict semantics for Until). Automata constructions typically maintain the formula $\varphi_1 \U \varphi_2$ as an obligation in the state or configuration that reads $a_i$. During the transition on $a_i$, the obligation $\varphi_1 \U \varphi_2$ at $i$ is transformed into one of two possible ways at $i+1$ non-deterministically: it either becomes an obligation $\varphi_2$ or a set of obligations $\{ \varphi_1, \varphi_1 \U \varphi_2 \}$. This construction is continued, leading to a set of obligations being maintained at each configuration along the word. Since the obligations are simply subformulas of the main formula being checked, and since there is no need to maintain multiple copies of the same subformula in a configuration, the state-space thus generated is finite and leads to an automaton construction. 

Let us now see what happens in the timed case. We consider the case $\varphi_1 \U_I \varphi_2$ where $I$ is a bounded interval of the form $[l, u]$, as this is the most complex case. Timing information needs to be encoded into the configurations maintaining the obligations. Let $w = (a_1, \tau_1) (a_2, \tau_2) \cdots$ be a timed word. Once again, we assume we need to assert $\varphi_1 \U_I \varphi_2$ at position $i$.

\subparagraph*{The 1-ATA method~\cite{ouaknine2007decidability}.} A method using 1-clock alternating timed automata (1-ATA) was proposed in \cite{ouaknine2007decidability}. On reading the letter $a_i$, the automaton generates an obligation of the form $(\varphi_1 \U_I \varphi_2, 0)$, where $0$ denotes the time since the obligation has been generated. On reading $a_{i+1}$ at the next position $i+1$, the obligation is non-determinitically transformed to either $(\varphi_2, 0)$ provided $\tau_{i+1} - \tau_i \in I$, or a set $\{(\varphi_1, 0), (\varphi_1 \U_I \varphi_2, \tau_{i+1} - \tau_i ) \}$.  In the second case, notice that there is a fresh obligation on $\varphi_1$ that is generated. When these transition rules are continued further down the run, more obligations on $\varphi_1$, with varying values of the associated timing component could be generated. Hence the configurations generated have unbounded ``width'': there is no bound on the number of obligations generated inside a configuration. A region-like abstraction was proposed in~\cite{ouaknine2007decidability} and it was shown that for \emph{finite timed words}, the abstraction along with a suitable entailment relation leads to a finite transition system. This was proposed in the context of full Metric Temporal Logic, and no customization for MITL was studied.

\subparagraph*{The MightyL approach~\cite{brihaye2013mitl,DBLP:conf/formats/BrihayeEG14,DBLP:conf/cav/BrihayeGHM17}.} Later methods concentrated on controlling the number of obligations generated per subformula when the logic is restricted to the MITL fragment. The first line of work in this direction is \cite{brihaye2013mitl} which developed an idea of 1-clock alternating timed automata with \emph{an interval semantics}. Intuitively, the interval semantics was used to abstract a set of obligations by considering the earliest and latest generated within that set. While~\cite{brihaye2013mitl} was restricted to finite words, it was extended to infinite words in~\cite{DBLP:conf/formats/BrihayeEG14}. These works laid the foundation for the tool MightyL~\cite{DBLP:conf/cav/BrihayeGHM17} which presented a full conversion from MITL to a network of timed automata, going via the abstraction described above. However, the number of discrete states generated for an until formula $\varphi_1 \U_I \varphi_2$ was exponential in a constant $\kappa = 1 + \lceil \dfrac{l}{u - l}\rceil$. A follow-up to MightyL, avoiding this blow-up, is the very recent work that we describe next.

\newcommand{\Pn}{\mathsf{Pn}}

\subparagraph*{The MightyPPL approach~\cite{DBLP:conf/tacas/HoKMMP26}.} This is a  different approach to MightyL where the novelty lies in encoding $\varphi_1 \U_I \varphi_2$ using a ``Pnueli operator'' $\Pn$: a formula $\Pn_I(\phi_1, \phi_2, \dots, \phi_n)$ holds at position $i$ if there exist $i < j_1 < j_2 < \cdots < j_n$ such that for each $1\le m \le n$, $\tau_{j_m} - \tau_i \in I$ and $(w, j_m) \models \phi_m$. Such a formula cannot be written in plain MITL. It is argued that $\varphi_1 \U_{(k, k+1)} \varphi_2$ is equivalent to the following formula (see Lemma~3 of~\cite{DBLP:conf/tacas/HoKMMP26}, Lemma 9 of~\cite{DBLP:journals/corr/abs-2510-01490}):
\begin{small}
\begin{align*}
 \varphi_1 \U_{(k, \infty)} \varphi_2 \land \bigvee_{i \in \{0, \dots, k+1\}} \left( \Pn_{[0,k+1)}(\underbrace{\phi^{\ge 1}, \phi^{\ge 1}, \cdots, \phi^{\ge 1}}_{i}, \varphi_2) \land \neg \Pn_{[0, k]}(\underbrace{\phi^{\ge 1}, \phi^{\ge 1}, \dots}_{i}, \phi^{\ge 1}) \right)
\end{align*}
\end{small}
where $\phi^{\ge 1} = \varphi_2 \land (\neg \varphi_2 \U_{[1, \infty]} \varphi_2)$. The above observation extends to intervals of the form $[l, u]$. Subsequently, the Pnueli operator $\Pn(\phi_1, \dots, \phi_n)$ is translated into a network of timed automata. Careful optimizations ensure that the additional blowup as in MightyL is avoided.

\subparagraph*{The TEMPORA approach~\cite{DBLP:journals/corr/abs-2510-14699, DBLP:conf/tacas/AkshayCGGS26}.} Another recent approach to MITL satisfiability uses the model of Generalized Timed Automata (GTA)~\cite{DBLP:conf/cav/AkshayGGJS23}. The novelty in this approach is the use of \emph{future clocks} (called \emph{timers} here) to guess times to certain good witnesses. Intuitively, a timer is assigned (non-deterministically) a non-negative value and then decreases towards $0$; technically and equivalently, a timer can be assigned a non-positive value and made to increase towards $0$ as in~\cite{DBLP:conf/cav/AkshayGGJS23}. When the timer hits $0$, we say it \emph{times out}. To assert $\varphi_1 \U_I \varphi_2$ at position $i$, this approach makes two guesses (their approach uses a weak Until semantics, which we reformulate for the strict semantics):
\begin{itemize}
\item the last position $j$ in the interval $\tau_i + I$ such that $\varphi_2$ is true at $j$, and $\varphi_1$ holds at all positions between $i+1$ and $j-1$; the time to this position $j$ from $i$, i.e., $\tau_j - \tau_i$ is stored in a future $x_1$. 
\item the first position $j'$ after $\tau_i + I$ such that $\varphi_2$ is true at $j'$ and $\varphi_1$ is true at all positions between $i+1$ and $j'-1$; the time to $j'$, i.e., $\tau_{j'} - \tau_{i}$ is stored in a future $y_1$.  
\end{itemize}
If such a position does not exist, the corresponding clock is undefined. Let us restrict to the case when both positions exist. As time elapses, the absolute values of these future clocks decrease. As long as $y_1$ has not timed out, it is shown that no new obligation needs to be generated, since for each intermediate position, either $j$ or $j'$ acts as a witness, or it can be ascertained that $\varphi_1 \U_I \varphi_2$ is not true at that encountered intermediate position. By construction, $y_1$ is initially set to a value greater that $I$. Therefore, the time elapse between its initial origin until its timeout is at least $|I|$. Hence, for $|I|$ time units, no new obligation needs to be generated. After that, fresh guesses $x_2, y_2$ in the same spirit as above are made. Hence, at any point of time during the run, there could be future clocks $x_1, y_1, x_2, y_2, \dots, x_i, y_i$ maintaining the time remaining to suitably guessed witnesses. Moreover, the number of such guesses is bounded. 
When these future clocks time out, or between the timeouts, the appropriate guesses need to be verified. For instance, between the timeout of each $x_j$ and $y_j$, formula $\varphi_1$ needs to be true, but $\varphi_2$ needs to be false. Secondly, there are subtleties arising due to witnesses coinciding: for example, the witness corresponding to $y_i$ may be the same as the one maintained by $x_{i+1}$. The final GTA carefully incorporates the guesses, the verification and the various subtle cases (Algorithm 2 in~\cite{DBLP:journals/corr/abs-2510-14699}, Algorithm 1 in~\cite{DBLP:conf/concur/0001G0S24}). 

\subparagraph*{Our approach.}  We roll back in spirit to the 1-ATA approach that maintains obligations of the form $(\varphi, 0)$, but in addition, we make use of the ability of timers from the GTA approach. In order to check if $(w, i) \models \varphi_1 \U_I \varphi_2$, while reading $a_i$ we make a guess that $\varphi_1 \U_I \varphi_2$ is true at $i$ by generating an obligation as follows:
\begin{align*}
(\varphi_1 \U_I \varphi_2, 0, t)
\end{align*}
where the first component denotes the formula that needs to be true, the second component denotes the \emph{age} of the obligation, i.e., the time since the obligation was generated and finally the third and most important component gives a \emph{waiting time}  $t$ predicting the time left to witness a position $j$ that satisfies $\varphi_2$. As expected, we require $t$ to be in the interval $I$.  

On a time elapse of $\delta$ units, the obligation becomes:
\begin{align*}
(\varphi_1 \U_I \varphi_2, \delta, t - \delta)
\end{align*}
The time elapse $\delta$ gets added to the age, and gets subtracted from the waiting time, as long as $\delta \le t$. Suppose the next letter $a_{i+1}$ occurs at this point $t_i + \delta$, the obligation changes as follows: if $t - \delta > 0$ (guessed position has not arrived yet), then the obligation continues to remain and alongside it generates an obligation for $\varphi_1$; if $t - \delta = 0$ (that is, waiting time is $0$), then there is a non-deterministic choice, either the current obligation is discharged and an obligation for $\varphi_2$ is generated denoting that a witness for $\varphi_2$ has been identified, or it continues to remain as it is, and generate one for $\varphi_1$. In the latter case, there can be no time elapse until the original $\varphi_1 \U_I \varphi_2$ obligation is discharged.
 
Figure~\ref{fig:until-obligation-idea} (left) gives an illustration: position $i$ is represented by the red dot, and position $j$ by the green dot. At $i$ an obligation $\theta_{\varphi}$ is generated which continues to remain until point $j$ where it is discharged, and a new obligation $\theta_{\varphi_2}$ is generated. At all intermediate positions where $\theta_{\varphi}$ is alive, an obligation for $\varphi_1$ is generated. 

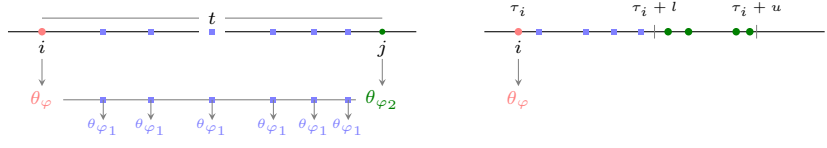
\begin{figure}
\centering 
\begin{tikzpicture}[scale=0.9]
\draw[thin] (-0.5,0) to (5.5, 0);
\node[circle, fill=red!50, inner sep=1pt] at (0,0) {};
\node[below] at (0,0) {\scriptsize $i$};

\node[circle, fill=green!50!black, inner sep=0.8pt] at (5,0) {};
\node[below] at (5,0) {\scriptsize $j$};

\draw[thin, gray] (0,0.2) -- node[midway, fill=white] {\scriptsize \textcolor{black}{$t$}} (5,0.2);

\foreach \x in {0.9, 1.6, 2.5, 3.4, 4, 4.5}
{\node[rectangle, fill=blue!50, inner sep=1pt] at (\x,0) {};}

\node [red!50] (0) at (0,-1) {\scriptsize $\theta_{\varphi}$};
\node [green!50!black] (1) at (5, -1) {\scriptsize $\theta_{\varphi_2}$};
\draw [gray, thin] (0) to (1);

\draw [very thin, gray, ->, >=stealth] (0,-0.4) -- (0, -0.8);
\draw [very thin, gray, ->, >=stealth] (5, -0.4) -- (5, -0.8);
\foreach \x in {0.9, 1.6, 2.5, 3.4, 4, 4.5}
{\node [blue!50] at (\x, -1.4) {\tiny $\theta_{\varphi_1}$};
\node[rectangle, fill=blue!50, inner sep=1pt] at (\x,-1) {};
\draw[very thin, gray, ->, >=stealth] (\x, -1) -- (\x, -1.3);}

\begin{scope}[xshift=7cm]
\draw[thin] (-0.5,0) to (4.5, 0);
\node[circle, fill=red!50, inner sep=1pt] at (0,0) {};
\node[below] at (0,0) {\scriptsize $i$};

\draw [very thin, gray] (2, 0.1) -- (2,-0.1);
\draw [very thin, gray] (3.5, 0.1) -- (3.5, -0.1);
\node [above] at (2, 0.1) {\tiny $\tau_i + l$};
\node [above] at (3.5, 0.1) {\tiny $\tau_i + u$};
\node [above] at (0, 0.1) {\tiny $\tau_i$};

\foreach \x in {0.3, 1, 1.4, 1.8}
{\node [rectangle, fill=blue!50, inner sep=1pt] at (\x,0) {};}

\foreach \x in {2.2, 2.5, 3.2, 3.4}
{\node [circle, fill=green!50!black, inner sep=1pt] at (\x,0) {};}

\node [red!50] at (0,-1) {\scriptsize $\theta_{\varphi}$};
\draw [very thin, gray, ->, >=stealth] (0,-0.4) -- (0, -0.8);
\end{scope}
\end{tikzpicture}
\caption{Illustrating the role of obligations for $\varphi = \varphi_1 ~\U_I~\varphi_2$ (left) and $\varphi = \varphi_1 ~\dual{\U}_I~ \varphi_2$ (right).}
\label{fig:until-obligation-idea}
\end{figure}

For Dual Until formulas, the role of obligations is different, mimicking the semantics. Consider a Dual Until formula $\varphi = \varphi_1 \dual{\U}_I \varphi_2$, with $I = [l, u]$ (say). To check if $(w, i) \models \varphi$, an obligation of the following form is generated:
\begin{align*}
 (\varphi_1\dual{\U}_I\varphi_2, 0, u)
\end{align*}
where the first two components are the same as before; and the third component simply maintains time until when this obligation needs to be alive. In Figure~\ref{fig:until-obligation-idea} (right), an obligation is shown to be generated at position $i$. The green dots inside the interval $[\tau_i + l, \tau_i + u]$ are the positions where the dual Until condition needs be checked. These positions can be identified by looking at the age of the obligation: when age is $\ge l$ we are inside the interval. 

On a time elapse of $\delta$, the Dual until obligation transforms to $(\varphi_1 \dual{\U}_I \varphi_2, \delta, u - \delta)$ if $\delta \le u$. When $\delta > u$, a time elapse of $\delta$ simply makes the obligation vanish. This is in sharp contrast with Until obligations, where time was not allowed to elapse beyond the waiting time that was initially guessed. For Until formulas, the obligation is guessing an event that should happen at $t$ time units and hence it is ``obliged'' to witness it, else it deadlocks. For Dual Until, the role is to ensure that upto a certain time, some events should happen. Beyond that, the obligation is useless and hence can be removed. 

In a nutshell: on each letter, an obligation is potentially generated, with a specific role, until it is discharged. When it is alive, the obligation may spawn more obligations depending on its role. In Section~\ref{sec:obligations} we will present a formal definition of a transition system made up of obligations, whose accepting runs are precisely the words where the formula is true. 
Notice that we do not predict special kinds of witnesses as in the GTA approach. Instead we just guess \emph{some witness}. The key question is how do we obtain an upper bound on the number of generated obligations with these mundane guesses. Section~\ref{sec:removing-redundant-obligations} discusses simple rules to reduce obligations that results in a uniform bound on the number of obligations that are stored.

\section{Obligations}
\label{sec:obligations}

Before we proceed further, we remark that in previous works on this subject, the word ``obligation'' has been used in its natural language sense to mean an assertion that should be checked. Moreover, obligations meant different things in different works. In our work, we consider another kind of obligations for which we provide a formal definition (Definition~\ref{def:obligation}).  
We distinguish between bounded and unbounded intervals: an interval is bounded if its right endpoint is in $\N$, and is unbounded when it is $\infty$. For formulas with unbounded intervals, we make the third component inactive in the obligations. This ensures that all waiting times that we use are globally bounded by a constant determined by the formula. This fact crucially helps in the symbolic computation explained in Section~\ref{sec:symbolic}. 

Fix an MTL formula $\varphi$ for the rest of this section. Indeed, all the results we develop in this section work for the full class of MTL.

\begin{definition}[Obligation]
\label{def:obligation} An \emph{obligation} for $\varphi$ is a triple $\theta = (\varphi_1 B_I \varphi_2,x,t)$ with $B \in \{ \U, \dual{\U} \}$ such that $\varphi_1 B_I \varphi_2$ is a subformula of $\varphi$, $x \in \Rpos$  and $t \in \Rpos$ when $I$ is bounded, and $t = \bot$ when $I$ is unbounded. Component $x$ stores the age of the obligation (the time since it was created) and $t$ stores the waiting time of the obligation -- the time until it gets removed; when $t = \bot$, it does not get removed. We write \emph{obligation-set} for a set of obligations, and use letter $\Theta$ to denote it.
\end{definition}

The initiation of an obligation depends on the form of $\varphi$. For instance consider $(a ~\U_{[1, 4]}~ b) \land ((\F_{[5, 10]} c)~\U_{[3,5]}~ d)$. There are two top-level Until operators $\U_{[1,4]}$ and $\U_{[3,5]}$. The initiation involves guessing witnesses for each of these two operators. Here is a possible guess $\{ (a~\U_{[1,4]}~b, 0, 3.7), ~((\F_{[5,10]} c)~\U_{[3,5]}~ d, 0, 4.8) \}$; another possible guess could have been $\{ (a~\U_{[1,4]}~b, 0, 1.5), ~((\F_{[5,10]} c)~\U_{[3,5]}~ d, 0, 3.3) \}$. 
The next definition formalizes this idea by associating to each formula $\varphi$ and a letter $a$, an initial set of possible obligation-sets, denoted as $\Ii_{\varphi, a}$. We require to make a distinction between a formula (like $\top$) which holds without any obligations, versus a formula (like $\bot$) which cannot be satisfied at all. When $\varphi$ cannot be satisfied, we leave $\Ii_{\varphi, a}$ undefined. When $\varphi$ is satisfied without any obligations, we define $\Ii_{\varphi,a}$ to be an empty set denoting that no obligations are needed. When we write $\Theta \in \Ii_{\varphi,a}$ we intrinsically assume that $\Ii_{\varphi,a}$ is defined.  

\begin{definition}[Initial set of obligation-sets]
    \label{def:initial-obligation-sets}
Initial set of obligation-sets $\Ii_{\varphi, a}$ of an MTL formula $\varphi$ w.r.t. a letter $a \in \Sigma$ is a set of obligation-sets, defined as follows:
\begin{itemize}
    \item when $\varphi = \top$, $\Ii_{\varphi, a} = \{ \}$; when $\varphi = \bot$, $\Ii_{\varphi, a}$ is undefined
\item when $\varphi = a$,  $\Ii_{\varphi,a} = \{\}$; when $\varphi = b \neq a$, $\Ii_{\varphi,b}$ is undefined
\item when $\varphi = \varphi_1 \land \varphi_2$, $\Ii_{\varphi, a} = \{\Theta_1 \cup \Theta_2 \mid \Theta_1 \in \Ii_{\varphi_1,a} \text{ and }\Theta_2\in \Ii_{\varphi_2,a}\}$
\item when $\varphi = \varphi_1 \lor \varphi_2$, then $\Ii_{\varphi,a} = \Ii_{\varphi_1,a} \cup \Ii_{\varphi_2,a}$
\item when $\varphi = \varphi_1 ~\U_{I}~ \varphi_2$, 
\[\Ii_{\varphi,a} = \begin{cases}
\{~\{(\varphi_1 \U_{I} \varphi_{2},0,t)\}\mid ~t \in I~\} & \text{ if }I \text{ is bounded } \\
\{~\{(\varphi_1 \U_{I} \varphi_{2},0,\bot)\}~\} & \text{ if }I \text{ is unbounded }
\end{cases}\]
\item when $\varphi = \varphi_1 ~\dual{\U}_{I}~ \varphi_2$, then 
\[\Ii_{\varphi,a} = \begin{cases}
\{~\{(\varphi_1 ~\dual{\U}_{I}~ \varphi_{2},0,u)\}~\} & \text{ if }I \text{ is bounded, $u$ is right end-point of $I$}  \\
\{~\{(\varphi_1 ~\dual{\U}_{I}~ \varphi_{2},0,\bot)\}~\} & \text{ if }I \text{ is unbounded }
\end{cases}\]
\end{itemize}
\end{definition}
Intuitively, each element of $\Ii_{\varphi,a}$ is a guess for why $\varphi$ is true at the position which reads $a$. 
As time evolves and more actions appear, the transitions are defined to verify this guess (or deadlock in case of wrong guesses). We now formalize how obligations evolve as time progresses and more letters are seen.
Each individual obligation may result in an obligation-set when a new letter is seen. 
Hence we first define transitions as a function from an obligation to an obligation-set, and then define how this extends as a function from an obligation-set to an obligation-set.

The delays behave differently for $\U_I$ and $\dual{\U}_I$ obligations: $\U_I$ obligations cannot elapse more than their waiting time, whereas $\dual{\U}_I$ obligations simply vanish when time elapses more than their waiting time, and hence do not block time from elapsing.

\begin{definition}[Delays on obligations]
    Let $\delta \in \Rpos$. 
    \begin{align*}
     (\varphi_1\U_I\varphi_2, x, t) & \xra{\delta} \{ (\varphi_1 \U_I \varphi_2, x + \delta, t - \delta) \} \qquad  \text{ when } \delta \le t \\
     (\varphi_1 \dual{\U}_I \varphi_2, x, t) & \xra{\delta} \begin{cases}
      \{ (\varphi_1 \dual{\U}_I \varphi_2, x + \delta, t - \delta) \} \qquad  & \text{ when } \delta \le t \\
      \{\} \qquad & \text{ when } \delta > t
     \end{cases} \\
     (\varphi, x, \bot) & \xra{\delta} \{ (\varphi, x + \delta, \bot)\}
    \end{align*}
\end{definition}
Notice the difference in the delay semantics between $\U$ and $\dual{\U}$. 

We move on to discrete transitions, starting with $\U$ obligations. As illustrated in Figure~\ref{fig:until-obligation-idea} (left), the role of an Until obligation for $\varphi = \varphi_1 \U_I \varphi_2$ is to keep generating obligations for $\varphi_1$, until its waiting time becomes $0$ when it non-deterministically decides to lance an obligation for $\varphi_2$ or keep continuing to check $\varphi_1$ without elapsing time. Lancing an obligation for a formula $\psi$ amounts to guessing an initial obligation-set for $\psi$.

When some $\Ii_{\psi,a}$ is not defined, the transition that involves obligation-sets from $\Ii_{\psi,a}$ is also not defined. Also notice that the transitions defined below are non-deterministic: not only because of the multiple transitions possible when the waiting time $t = 0$, but also because we are non-deterministically picking an obligation-set from initial obligation-sets when applicable.
\begin{definition}[Discrete transitions for Until obligations.] 
    \label{def:until-obligations-discrete}
Let $a \in \Sigma$. 
\[(\varphi_1 \U_{I} \varphi_{2}, x, t) \xra{a}\begin{cases}
    \{(\varphi_1 \U_{I} \varphi_{2}, x, t)\} \cup \Theta_{1} & \text{ for  } \Theta_{1} \in \Ii_{\varphi_1,a} \quad \text{ when }t \geq 0\\
    \Theta_{2} & \text{ for }\Theta_2 \in \Ii_{\varphi_2,a} \quad \text{ when }t = 0 \\
\end{cases}\]
\[(\varphi_1 \U_{I} \varphi_{2}, x, \bot) \xra{a} \begin{cases}
    \{(\varphi_1 \U_{I} \varphi_{2}, x, \bot)\} \cup \Theta_{1} & \text{ for  } \Theta_{1} \in \Ii_{\varphi_1,a} \\
    \Theta_{2} & \text{ for }\Theta_2 \in \Ii_{\varphi_2,a} \quad \text{ when }x \in I 
\end{cases}\]
\end{definition}

For the Dual Until, let us refer to Figure~\ref{fig:until-obligation-idea} (right). For all the green points inside the interval the obligation needs to check that either $\varphi_2$ is true, or there is an earlier point where $\varphi_1$ is true. We can decompose this reasoning as follows: 
\begin{itemize}
\item either $\varphi_1$ is guessed to be true already before the interval begins (the blue rectangles in the figure), in which case the obligation can be discharged after generating one for $\varphi_1$,
\item or the obligation has entered the interval and starts to check $\varphi_2$ at each position; if additionally, $\varphi_1$ is guessed to be true at some position inside the interval, then the obligation can be discharged.
\end{itemize}
We use the notation $x < I$ to mean the value of $x$ is strictly less than every point in the interval. Let $l$ be the left end-point of $I$. We will write $l \triangleleft_I x$ to indicate that $l \le x$ when $I$ is left-closed, and $l < x$ when $I$ is left open (essentially, the lower bound of $I$ is satisfied by $x$).

\begin{definition}[Discrete transitions for dual Until obligations]
    \label{def:dual-until-discrete}
Let $a \in \Sigma$, and $t \in \Rpos \cup \{\bot\}$. 
\[(\varphi_1 \dual{\U}_{I} \varphi_{2}, x, t) \xra{a} \begin{cases}
    \{(\varphi_1 \dual{\U}_{I} \varphi_{2}, x, t)\} & \text{ when }x < I \\
    \Theta_1 & \text{ for }\Theta_1 \in \Ii_{\varphi_1,a} \quad \text{ when }x < I \\
    \{(\varphi_1 \dual{\U}_{I} \varphi_{2},x,t)\} \cup \Theta_2 & \text{ for }\Theta_2 \in \Ii_{\varphi_2,a} \quad \text{ when }l \triangleleft_I x \\
    \Theta_1 \cup \Theta_2 & \text{ for }\Theta_1 \in \Ii_{\varphi_1,a},\Theta_2 \in \Ii_{\varphi_2,a} \quad \text{ when } l \triangleleft_I x
\end{cases}\]
\end{definition}

For an obligation $\theta$, we will write $\theta \xra{\delta, a} \Theta'$ in short for a sequence of a delay followed by action: $\theta \xra{\delta} \{\theta'\}$ and $\theta' \xra{a} \Theta'$. We now extend the definition of transitions to obligation-sets.

\begin{definition}[Transitions on obligation-sets] Let $\varphi$ be an MTL formula, and let $\Theta, \Theta'$ be obligation-sets for $\varphi$. For a delay $\delta \in \Rpos$ and a letter $a \in \Sigma$, we say that $\Theta'$ is a $(\delta, a)$-extension of $\Theta$ if for every $\theta \in \Theta$, there exists $\Theta'' \incl \Theta'$ such that $\theta \xra{\delta, a} \Theta''$. We say that $\Theta'$ is a \emph{minimal} $(\delta, a)$-extension if no strict subset $\Theta'' \subset \Theta'$ is a $(\delta,a)$-extension. Finally, we define $\Theta \xra{\delta, a} \Theta'$ if $\Theta'$ is a minimal $(\delta, a)$-extension of $\Theta$.
\end{definition}

Minimal extensions ensure that the resulting obligation-sets do not contain two different guesses at a position for the same formula. This motivates us the following definition which generalizes this idea.

\begin{definition}[Well-formed obligation-set]
An obligation-set $\Theta$ is \emph{well-formed} if no two obligations over the same formula have the same age: for every pair $(\psi, x_1, t_1), (\psi, x_2, t_2) \in \Theta$, $t_1 \neq t_2$ implies $x_1 \neq x_2$. 
\end{definition}

\begin{restatable}{lemma}{wellformed}\label{lem:minimal-extension-guarantess-well-formedness}
Let $\Theta$ be a well-formed obligation-set. Then, for every $\delta \in \Rpos$ and $a \in \Sigma$, every minimal $(\delta, a)$-extension of $\Theta$ is well-formed as well. 
\end{restatable}

These definitions lead us to the final transition system of obligations. Nodes in this transition system contain well-formed obligation-sets.  Infinite runs in this transition system correspond to infinite words, and the node visited after position $i$ in the word represents the obligations at $i$. Figure~\ref{fig:obligation-runs} depicts two formulas and a run of the obligation graph for each of them. 
The example below gives an idea of how a run evolves on a word, after which we define the obligation graph and its accepting condition more formally.
\begin{figure}
    \centering
    \scalebox{0.8}{
    \begin{tikzpicture}
        \draw [gray] (6, 3.5) -- (6,-7.5);
        
    \begin{scope}[scale=0.8, yshift = 2cm,xshift=-1.5cm]
    \draw (0,0) to (8,0);
    \foreach \i in {0, 1, 2, 3, 4, 5, 6, 7}
    {
    \draw [gray] (\i, -0.1) -- (\i, 0.1);
    \node at (\i, -0.25) {\scriptsize \i};
    }
    \foreach \i in {0, 1.8, 2, 7}{
        \node [circle, inner sep=1pt, fill=red!50] at (\i, 0) {};
        \node at (\i, 0.25) {\scriptsize $b$};
    }
    \foreach \i in {0.3, 1.5, 2.4, 2.9, 3.5}{
        \node [circle, inner sep=1pt, fill=blue!50] at (\i, 0) {};
        \node at (\i, 0.25) {\scriptsize $a$};
    }
    \node at (3.8, 1) {\scriptsize $w = (b, 0) (a, 0.3) (a, 1.5) (b, 1.8) (b, 2) (a, 2.4) (a, 2.9) (a, 3.5) (b, 7) \cdots$};
    \node at (4, 1.7) {\scriptsize $\varphi = \varphi_1 \U_{[2,10]} \varphi_2$};
    \node at (4, 2.3) {\scriptsize $\varphi_1 = \F_{[1,3]} a$ \quad $\varphi_2 = b$};
    \end{scope} 

    \begin{scope}[yshift=-2cm,xshift=1.5cm,scale=0.8]
        \node (0) at (0.5, 3.5) {\scriptsize $\{\}$};
        \node (1) at (0.5, 2.5) {\scriptsize $\{(\varphi,0,2)\}$};        
        \node (2) at (0.5, 1) {\scriptsize $\{(\varphi,0.3,1.7),(\varphi_1,0,2.1)\}$};        
        \node (3) at (0.5, -0.25) {\scriptsize $\{ (\varphi, 1.5, 0.5), (\varphi_1, 1.2, 0.9), (\varphi_1, 0, 1.4)\}$}; 
        \node (4) at (0.5, -1.75) {\scriptsize $\{(\varphi, 1.8, 0.2), (\varphi_1, 1.5, 0.6), (\varphi_1, 0.3, 1.1), (\varphi_1, 0, 1.7) \}$};
        \node (5) at (0.5, -3) {\scriptsize $\{(\varphi_1, 1.7, 0.4), (\varphi_1, 0.5, 0.9), (\varphi_1, 0.2, 1.5) \}$};
        \node (6) at (0.5, -4.25) {\scriptsize $\{(\varphi_1, 0.9, 0.5), (\varphi_1, 0.6, 1.1)\}$};
        \node (7) at (0.5, -5.5) {\scriptsize $\{ (\varphi_1, 1.1, 0.6)\}$};
        \node (8) at (0.5, -6.75) {\scriptsize $\{\}$};

        \node (l1) [purple] at (0.5,2.2) {\tiny $\varphi.1$};
        \node (l21) [purple] at (-0.5,0.7) {\tiny $\varphi.1$}; 
        \node (l22) [purple] at (1.5,0.7) {\tiny $\varphi_1.1$};
        \node (l2) [] at (0.5,0.7) {};
        \node (l31) [purple] at (-1.5,-0.55) {\tiny $\varphi.1$};
        \node (l32) [purple] at (0.5,-0.55) {\tiny $\varphi_1.1$};
        \node (l33) [purple] at (2.5,-0.55) {\tiny $\varphi_1.2$};
        \node (l41) [purple] at (-2.5,-2.05) {\tiny $\varphi.1$};
        \node (l42) [purple] at (-0.5,-2.05) {\tiny $\varphi_1.1$};
        \node (l43) [purple] at (1.5,-2.05) {\tiny $\varphi_1.2$};
        \node (l44) [purple] at (3.5,-2.05) {\tiny $\varphi_1.3$};
        \node (l4) [] at (0.5,-2.05) {};
        \node (l51) [purple] at (-1.6,-3.3) {\tiny $\varphi_1.1$};
        \node (l52) [purple] at (0.5,-3.3) {\tiny $\varphi_1.2$};
        \node (l53) [purple] at (2.5,-3.3) {\tiny $\varphi_1.3$};
        \node (l61) [purple] at (-0.5,-4.55) {\tiny $\varphi_1.2$};
        \node (l62) [purple] at (1.5,-4.55) {\tiny $\varphi_1.3$};
        \node (l6) [] at (0.5,-4.55) {};
        \node (l7) [purple] at (0.5,-5.8) {\tiny $\varphi_1.3$};
    \end{scope}
    \begin{scope}[->, >=stealth]
        \draw (0) to node [right] {\tiny $0, b$} (1);
        \draw (l1) to node [right] {\tiny $0.3, a$} (2);
        \draw (l2) to node [right] {\tiny $1.2, a$} (3);
        \draw (l32) to node [right] {\tiny $0.3, b$} (4);
        \draw (l4) to node [right] {\tiny $0.2, b$} (5);
        \draw (l52) to node [right] {\tiny $0.4, a$} (6);
        \draw (l6) to node [right] {\tiny $0.5, a$} (7);
        \draw (l7) to node [right] {\tiny $0.6, a$} (8);
        
    \end{scope}
    \begin{scope}[scale=0.8, yshift = 2cm,xshift=9.5cm]
    \draw (0,0) to (6,0);
    \foreach \i in {0, 1, 2, 3, 4, 5}
    {
    \draw [gray] (\i, -0.1) -- (\i, 0.1);
    \node at (\i, -0.25) {\scriptsize \i};
    }
    \foreach \i in {0.2}{
        \node [circle, inner sep=1pt, fill=red!50] at (\i, 0) {};
        \node at (\i, 0.25) {\scriptsize $b$};
    }
    \foreach \i in {0, 1, 2.6, 3.1, 3.5, 5}{
        \node [circle, inner sep=1pt, fill=blue!50] at (\i, 0) {};
        \node at (\i, 0.25) {\scriptsize $a$};
    }
    \node at (3, 1) {\scriptsize $w' = (a, 0) (b, 0.2) (a, 1) (a, 2.6) (a, 3.1) (a, 3.5) (a, 3.5) (a, 5) \cdots$};
    \node at (3, 1.7) {\scriptsize $\psi = \psi_1 \U_{[2,6]} \psi_2$};
    \node at (3, 2.3) {\scriptsize $\psi_1 = b~\dual{\U}_{[2,3]} a$ \quad $\psi_2 = a$};
    \end{scope}
    \begin{scope}[yshift=-2cm,xshift=9cm,scale=0.8]
        \node (0') at (1, 3.5) {\scriptsize $\{\}$};
        \node (1') at (1, 2.5) {\scriptsize $\{(\psi,0,2.6)\}$};        
        \node (2') at (1, 1) {\scriptsize $\{(\psi,0.2,2.4),(\varphi_1,0,3)\}$};        
        \node (3') at (1, -0.25) {\scriptsize $\{ (\psi, 1, 1.6),(\psi_1, 0.8, 2.2), (\psi_1, 0, 3)\}$}; 
        \node (4') at (1, -1.45) {\scriptsize $\{(\psi_1, 2.4, 0.6), (\psi_1, 1.6, 1.4)\}$};
        \node (5') at (1, -2.65) {\scriptsize $\{(\psi_1, 2.9, 0.1), (\psi_1, 2.1, 0.9)\}$};
        \node (6') at (1, -3.65) {\scriptsize $\{{(\psi_1, 2.5, 0.5)}\}$};
        \node (7') at (1, -4.65) {\scriptsize $\{\}$};

        \node (l1') [purple] at (1,2.2) {\tiny $\psi.1$};
        \node (l2') [purple] at (0.2,0.7) {\tiny $\psi.1$};
        \node (l3') [purple] at (-0.7,-0.55) {\tiny $\psi.1$};
    \end{scope}
    \begin{scope}[->, >=stealth]
        \draw (0') to node [right] {\tiny $0, a$} (1');
        \draw (l1') to node [right] {\tiny $0.2, b$} (2');
        \draw (2') to node [right] {\tiny $0.8, a$} (3');
        \draw (3') to node [right] {\tiny $1.6, a$} (4');
        \draw (4') to node [right] {\tiny $0.5, a$} (5');
        \draw (5') to node [right] {\tiny $0.4, a$} (6');
        \draw (6') to node [right] {\tiny $1.5, a$} (7');
    \end{scope}

    \end{tikzpicture}
    }
    \caption{Examples on the obligations generated for two MITL formulas}
    \label{fig:obligation-runs}
\end{figure}
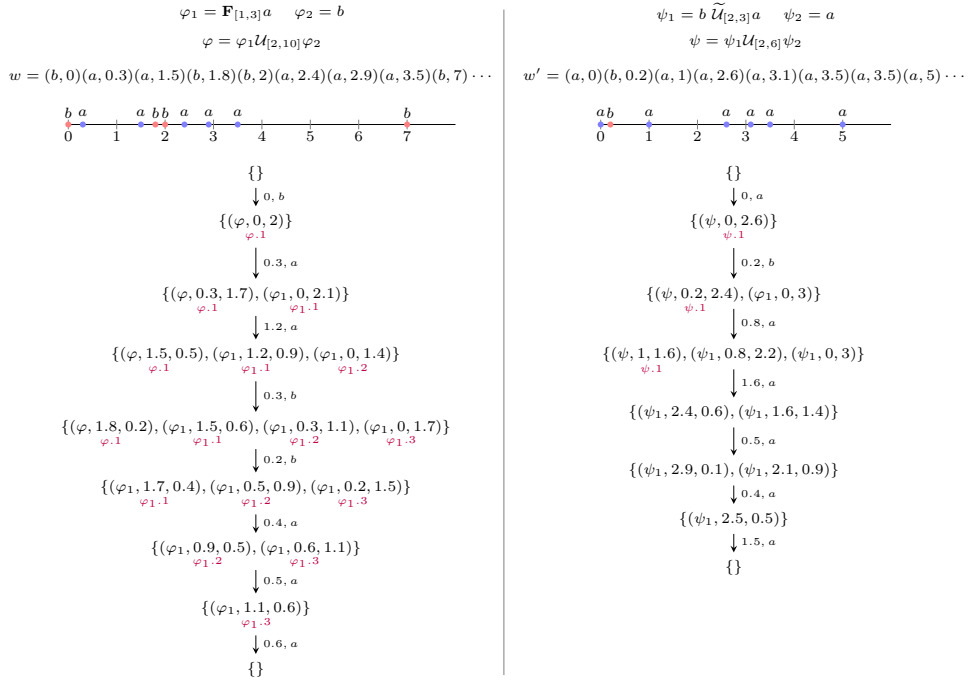

\begin{example}\label{ex:obligation-graph}
    First, we look at $\varphi = (\F_{[1,3]} a) \U_{[2,10]} b$ and the run in $\ob(\varphi)$ for a word $w$, where the $\id$ for each obligation-set is shown below the respective obligations. We observe that $(w,1) \sat \varphi$. One reason why the formula holds at position $1$ is that there is  $b$ at position $5$ of $w$ (hence $(w, 5) \sat \varphi_2$) such that the time elapsed between positions $1$ and $5$ falls in the interval $[2, 10]$, and for each position $(w, i)$ for $i \in \{2, 3, 4\}$, we have $(w, i) \models \F_{[1,3]} a$. Notice that, there are multiple witnesses for $(w, 2) \models \varphi_1$: $(a, 1.5)$, $(a, 2.4)$ and $(a, 2.9)$. One such accepting run is depicted in the figure. Here, on reading $(b,0)$, we make a guess that $\varphi_2$ is satisfied $2$ time units from now, thus generating the obligation $(\varphi, 0, 2)$. After a time elapse of $0.3$ an $a$ is seen. The obligation $(\varphi, 0, 2)$ now has evolved to $(\varphi, 0.3, 1.7)$. In order to satisfy $\varphi_1$ at this position, we generate a new obligation: we predict that $a$ will hold $2.1$ time units from now, and thus add $(\varphi_1,0,2.1)$ to the current pool of obligations. In the next transition: a $1.2$ time elapse after which an $a$ is seen, the existing obligations evolve to $(\varphi,1.5,0.5)$ and $(\varphi_1,1.2,0.9)$ respectively, and because we need to satisfy $\varphi_1$ at this new position, a new obligation $(\varphi_1,0,1.4)$ is generated and added to the pool. The run progresses this way, until we read position $5$ (i.e. the fifth transition in the figure), where on time elapse of $0.2$, the first generated obligation would have evolved to $(\varphi,2,0)$, and we will release this obligation while generating an obligation for $\varphi_2$, which is $\{\}$ in this case. The various $\varphi_1$ obligations generated along the run are released in a similar way, on reading an $a$ after the waiting time has evolved to $0$.
Another possible run for this word could have been guessing the $(b, 7)$ to be the witness for $\varphi_2$ and generating an obligation $(\varphi, 0, 7)$ on the first transition. Here, when we would have read $(a, 2.9)$, we still would not have seen this $b$ yet, and so would need to make $\F_{[1,3]} a$ true. Hence, we would have had to make some guess for an $a$ to occur between $3.9$ and $5.9$ time units. For this particular word $w$, this would lead to a deadlock: after $(a, 3.5)$, it would not be possible to read $(b, 7)$ since this time elapse would have been more than the guess that we have made. 

Looking at the second example, an MITL formula $\psi = (b~\dual{\U}_{[2,3]}a)~\U_{[2,6]}a$ and its run for the word $w'$, we see that $w',1 \sat \psi$. One reason for this is that there is an $a$ at position $4$ of $w'$ and the time distance between the positions falls in $I = [2,5]$, and $w,i \sat \varphi_1$ for positions $i \in \{2,3\}$. The latter is because for both positions, the events between time elapse $[2,3]$ satisfy $a$. This is the accepting run depicted in the figure (right). The outermost transition is still an until, so the way the obligations are generated is similar to the first example. The difference is the dual until obligations are also generated (in transitions two and three). Notice that both the created obligations  remain in the following transitions until enough time elapses making the waiting time go below $0$, at which time they get removed (in transitions six and seven respectively). 
\end{example}

Accepting runs are those where every Until obligation that is generated is eventually discharged.  To implement this accepting condition, we add identifiers to every Until obligation and define a generalized B\"uchi acceptance condition based on these identifiers. In Figure~\ref{fig:obligation-runs}, obligations for formulas $\varphi$, $\varphi_1$, and $\psi$ are annotated below with the formula followed by a number. These are the identifiers of the respective obligations. The journey of an obligation can be conveniently tracked using its identifier. For example, we can see that $\varphi_1 \cdot 1$ is created at position $2$ of the word, and travels until it is discharged at position $6$.

\begin{definition}[Types and Identifiers] For an obligation $\theta = (\varphi, x, t)$, we define $\type(\theta) = \varphi$. For every Until formula $\psi$, we define a set  $\identifiers_\psi := \{ (\psi \cdot i) \mid i \in \N \}$. Given an obligation-set $\Theta$, a \emph{naming function} $\id$ for $\Theta$ is an injection that maps every Until obligation $\theta \in \Theta$ to an element in $\identifiers_{\type(\theta)}$. Given a naming function $\id$ for $\Theta$ and an Until formula $\psi$, we define $\range(\id, \psi) = \{ i \in \N \mid \text{ some obligation $\theta \in \Theta$ with type $\psi$ is mapped to } \psi \cdot i \text{ via } \id\}$.
\end{definition}

In the definition below, we make use of these notations: (1) given an obligation-set $\Theta$ and a $\delta \ge 0$ s.t. all waiting times of obligations in $\Theta$ are at least $\delta$, we define $\Theta + \delta = \{ \theta + \delta \mid \theta \in \Theta \}$; (2) given an obligation $\theta = (\psi, x, t)$ s.t. $x \ge \delta$, we define $\theta - \delta:= (\psi, x - \delta, t + \delta)$. In the obligation graph below, we choose a natural naming convention where each fresh obligation is assigned the smallest available identifier of its type. 

\begin{definition}[Obligation graph] 
    \label{def:obligation-graph}
    For an MTL formula $\varphi$, we define a transition system called $\ob(\varphi)$. Nodes of $\ob(\varphi)$ are pairs $(\Theta, \id)$ consisting of a \emph{well-formed} obligation-set $\Theta$ and a naming function $\id$ for $\Theta$. There is a special node $\init_\varphi$ which is the initial node, and a special node $(\Theta_{\emptyset}, \id_{\emptyset})$ with $\Theta_{\emptyset} = \{\}$ representing the empty set of obligations. 

    For all $\delta \in \Rpos$ and $a \in \Sigma$ we add $\init_{\varphi} \xra{\delta, a} (\Theta,\id)$ for every $\Theta \in \Ii_{\varphi, a}$, with $\id(\theta) = \psi \cdot 1$ for every $\psi \in \type(\Theta')$ s.t. $\psi$ is an Until formula (note that there is at most one obligation of each type, since $\Theta$ is well-formed). Similarly, we add $(\Theta_{\emptyset}, \id_{\emptyset}) \xra{\delta, a} (\Theta_{\emptyset}, \id_{\emptyset})$. For two nodes $(\Theta, \id), (\Theta', \id')$ we add $(\Theta, \id) \xra{\delta, a} (\Theta', \id')$ if $\Theta \xra{\delta, a} \Theta'$ and the naming function $\id'$ in the target satisfies the following. 

    For every $\psi \in \type(\Theta')$ s.t. $\psi$ is an Until formula, let $i_\psi \in \N$ be the smallest number that is not used in $\Theta$, that is, $i_\psi \notin \range(\id, \psi)$.
    \begin{itemize}
        \item For every $\theta$ of type $\psi$ present in both $(\Theta + \delta)$ and $\Theta'$, $\id'(\theta) = \id(\theta - \delta)$. 
        \item As $\Theta'$ is well-formed, there is atmost one new obligation $\theta \in \Theta' \setminus (\Theta + \delta)$ s.t.  $\type(\theta)= \psi$. Then $\id'(\theta) = \psi \cdot i_\psi$. 
    \end{itemize}
\end{definition}

The acceptance condition essentially mimicks the fact that every $\U_I$ obligation that is generated is eventually discharged. When a new Until obligation arrives, a new identifier, say $\psi \cdot i$ is assigned to it. When it is discharged, $\psi \cdot i$ becomes inactive (notice the way we have defined $\id'$). Therefore, the acceptance condition amounts to verifying whether each identifier becomes inactive infinitely often. For Dual Until obligations, there is nothing to check. 

\begin{definition}[Acceptance condition]\label{def:acc-condition}
The acceptance condition for $\ob(\varphi)$ is defined as follows. For every Until subformula $\psi$ of $\varphi$, and for every $i \in \N$,  define $F_\psi^i := \{ (\Theta, \id) \mid i \notin \range(\id, \psi)\}$: $F_\psi^i$ contains the set of all nodes where $\psi \cdot i$ is not an active identifier. The (infinite) set $\Acc(\varphi) = \{ F^i_\psi \mid i \ge 0, \text{ and } \psi \text{ is an Until subformula of $\varphi$} \}$ describes the generalized B\"uchi acceptance condition for the obligation graph $\ob(\varphi)$.   
\end{definition}

For a timed word $w = (a_1, \tau_1) (a_2, \tau_2) \cdots$, a run of $w$ on $\ob(\varphi)$ from a node $(\Theta, \id)$ is an infinite sequence $(\Theta, \id) \xra{\delta_1, a_1} (\Theta_1, \id_1) \xra{\delta_2, a_2} \cdots$ where $\delta_i = \tau_i - \tau_{i-1}$ for all $i \ge 1$, and taking $\tau_0 = 0$. Since $\ob(\varphi)$ is non-deterministic, there could be several runs for $w$. A run is said to be \emph{accepting} if it satisfies the generalized B\"uchi condition given by $\Acc(\varphi)$: for every Until subformula $\psi$ and every $i \ge 0$, $F_\psi^i$ should be visited infinitely often. In other words, every time there appears an Until obligation with identifier $\psi \cdot i$, there should be a later point in the run where it disappears. 
We say a word $w$ is accepting from a node if it has some accepting run in $\ob(\varphi)$ from this node, and define the language of a node as the set of accepting words from that node. For the graph $\ob(\varphi)$, we define the language of $\ob(\varphi)$ as $L(\ob(\varphi)) = L(\init_{\varphi})$.

\begin{restatable}{theorem}{oblcorr}\label{thm:oblcorr}
Given an MTL formula $\varphi$ and a timed word $w$, we have $(w, 1) \sat \varphi$ iff $w \in L(\ob(\varphi))$.
\end{restatable}

\section{Reduced obligation graph}
\label{sec:removing-redundant-obligations}

When there are multiple obligations of the same type in a node of $\ob(\varphi)$ we can eliminate or merge some of them.

In this section, we will provide optimization rules that will ensure that there is a bound $k_{\psi}$ such that every node has at most $k_\psi$ obligations for a given MITL formula $\psi$, while still preserving the language of the obligation transition system. The optimization rules depend on the type of the obligations, whether they are $\U_I$ or $\dual{\U}_I$, and within each category, whether $I$ is bounded or not. 

\subparagraph*{Until with a bounded interval.}  We start with the case of an Until formula with a bounded interval $I$ which is typically the most difficult in all existing constructions. Let $l, u$ be the left and right endpoints of $I$. Let $\Theta$ be a obligation-set containing the following two obligations, one of them having been freshly created with age $0$:
\begin{align*}
\theta_1 = (\varphi_1 \U_I \varphi_2, x_1, t_1) \qquad \theta_2 = (\varphi_1 \U_I \varphi_2, 0, t_2)
\end{align*}
From the semantics of obligations, there was a position $i_1$ earlier in the word which generated $\theta_1$. The time that has elapsed ever since is $x_1$, and the time left for the witness to appear is $t_1$. Let $i_2 > i_1$ be the position in the word when $\theta_2$ was generated. Firstly, note that $t_2 \in I$ since the initial guess made should lie in $I$. The next two simple rules exhibit when we can keep only one obligation.

\begin{description}
\item[\removenew] if $t_1 \in I$, remove $\theta_2$.
\item[\modifyold] else, if $x_1 + t_2 \in I$, update $\theta_1$ to $(\varphi_1 \U_I \varphi_2, x_1, t_2)$ (the waiting time has been modified from $t_1$ to $t_2$) and remove $\theta_2$.
\end{description}
Intuitively, when $t_1 \in I$, the witness guessed by $\theta_1$ at position $i_1$ is already a witness for position $i_2$, and hence the presence of obligation $\theta_1$ is sufficient for the satisfaction of $\varphi_1 \U_I \varphi_2$ at position $i_2$. In some sense, $\theta_1$ covers for $\theta_2$.
For the \modifyold\ case: when $x_1 + t_2 \in I$, observe that the witness guessed by $\theta_2$ is a witness for position $i_1$ as well. Therefore it is unnecessary to maintain two different guesses and hence we modify $\theta_1$ to a new $\theta_1'$ that points to the guess made by $\theta_2$. When $\theta'_1$ is finally discharged, we verify the guesses made both at $i_1$ and at $i_2$: both $\theta_1$ and $\theta_2$ are covered by $\theta'_1$. 

Here is a technical point crucial for the correctness of the approach. Suppose $\theta_1$ covers $\theta_2$ leading to \removenew\ rule applied to eliminate $\theta_2$. Further down the line, it may so happen that $\theta_1$ is modified to obligation $\theta'_1$. It turns out that this new $\theta'_1$ also covers the previously removed obligation $\theta_2$, and hence does not contradict our previous decision to remove $\theta_2$. We illustrate this idea in the figure below. Assume $\theta_1$ is generated at point $i_1$ and it guesses a witness at point $j_1$ inside the required interval. When $i_2$ is encountered, \removenew\ is applied. Hence $j_1$ continues to be a witness for $i_2$. Now, when $i_3$ is seen, \modifyold\ rule is applied to ascertain that $j_3$ is a witness for $i_1$ and hence $\theta_1$ is updated to point to $j_3$. We claim that $j_3$ is a witness for $i_2$ as well: since $j_3$ satisfies the upper bound timing constraint wrt $i_1$, it satisfies the upper bound constraint wrt $i_2$; since $j_3$ satisfies lower bound constraint wrt $i_3$, we deduce that it satisfies the lower bound constraint wrt $i_2$. 
\begin{center}
\begin{tikzpicture}[scale=0.8]
\draw[thick, gray!50!white] (0,0) to (11,0);
\node[fill, circle, inner sep=1pt, red!50] at (1,0) {};
\node[below] at (1,0) {\scriptsize $i_1$};
\node[fill, circle, inner sep=1pt, red!50] at (2,0) {};
\node[below] at (2,0) {\scriptsize $i_2$};
\node[fill, circle, inner sep=1pt, red!50] at (3,0) {};
\node[below] at (3,0) {\scriptsize $i_3$};

\node at (6.5, 0) {$[$};
\node at (10.5, 0) {$\mathbf ]$};

\node[fill, circle, inner sep=1pt, green!50!black] at (7.5, 0) {};
\node[below] at (7.5, 0) {\scriptsize $j_1$};
\node[fill, circle, inner sep=1pt, green!50!black] at (9.5, 0) {};
\node[below] at (9.5, 0) {\scriptsize $j_3$};
\end{tikzpicture}
\end{center}

\emph{Bound on the number of obligations when $I$ is non-singular.} Let us now see how the two rules ensure a bound on the number of $\U_I$ obligations. For simplicity, assume $I = [l,u]$, a closed interval, and let $\psi = \varphi_1 \U_I \varphi_2$. The key observation that helps in getting a bound for MITL is the following. When neither \removenew\ nor \modifyold\ is applicable on $\theta_1= (\psi,x_1,t_1)$ and $\theta_2 = (\psi,0,t_2)$, we will have:
\begin{align}\label{eq:until-elimination-keep-both}
t_1 < l \qquad \text{ and } \qquad x_1 + t_2 > u
\end{align}
This is because: initially all timers are initiated with a value in $[l,u]$. As $t_1 \notin I$ (due to \removenew\ not being applicable), sufficient time has elapsed since it was generated and its value satisfies $t_1 < l$. Similarly if $x_1 + t_2 \notin I$ (due to \modifyold\ not being applicable) it means $x_1 + t_2 > u$ since we know that $t_2 \in I$ and $x_1 \ge 0$. 
Now, consider an obligation-set that has three obligations of type $\psi$, say $\theta_1,\theta_2,\theta_3$ for $\theta_i = (\psi,x_i,t_i)$. Assuming $x_1 > x_2 > x_3$, we see that when the second obligation was newly created, the obligation-set had two obligations $\theta_1 - x_2 = (\psi,x_1-x_2,t_1+x_2)$ and $\theta_2 - x_2 = (\psi,0,t_2+x_2)$, and as they were both kept in the obligation-set, by \eqref{eq:until-elimination-keep-both} we know that  
\begin{align*}
t_1 + x_2 < l \quad \text{ and } \quad (x_1 - x_2) + (t_2 + x_2) = x_1 + t_2 > u
\end{align*}
Similarly, we see that when the third obligation was newly created, the obligation-set contained $\theta_2-x_3 = (\psi,x_2-x_3,t_2+x_3)$ and $\theta_3-x_3 = (\psi,0,t_3+x_3)$, and by \eqref{eq:until-elimination-keep-both}:
\begin{align*}
t_2 + x_3 < l \quad \text{ and } \quad x_2 + t_3 > u
\end{align*}
Combining $x_1 + t_2 > u$ and $t_2 + x_3 < l$ gives
\begin{align*}
x_1 - x_3 > u- l
\end{align*} 

In summary: \emph{after the first two obligations for $\varphi_1 \U_I \varphi_2$, a third one needs to be stored only after $|I|$ time units}.  By this time the waiting time for the first obligation has come down to $< I$. This allows us to deduce that while the first obligation is still alive, at most $2 \times \lceil \frac{u}{u-l} \rceil$ obligations will be stored. Rewriting $u = l + (u - l)$ and $\kappa = \lceil \frac{l}{u-l} \rceil + 1$, we see that atmost $2 \kappa$ obligations for $\varphi_1 \U_{I} \varphi_2$ need to be stored in any node. 

\subparagraph*{Until with an unbounded interval.}
Now consider the case when $I$ is unbounded, that is of the form $[l, \infty)$ or $(l, \infty)$. An Until obligation $(\varphi_1 \U_I \varphi_2, x, \bot)$ with $I$ unbounded gets discharged at a point where $l \le x$ (when $I$ is left-closed) and $l < x$ (when $I$ is left-open). Suppose we have two obligations active in some obligation-set:
\begin{align*}
\theta_1 = (\varphi_1 \U_I \varphi_2, x_1, \bot) \qquad \text{ and } \qquad \theta_2 = (\varphi_1 \U_I \varphi_2, x_2, \bot) \qquad \text{ s.t. } \qquad x_1 > x_2
\end{align*}
Clearly, if $x_2 \in I$, then $x_1 \in I$ too. Therefore a position where $\theta_2$ finds a witness is also a witness position for $\theta_1$. Hence we are tempted to  remove $\theta_1$. This way, we store only a single obligation for Until with an unbounded interval. However there is a problem if we do this repeatedly: if we continue to remove the older obligation and store only the newest one, there may never be a position where an obligation is actually discharged. We may continue to keep postponing the obligations. To eliminate this situation, we adopt a different strategy where we store two obligations.

Let $\psi = \varphi_1 \U_I \varphi_2$ with $I$ unbounded. Suppose $\theta_1$ is the first obligation of type $\psi$ that appears in the run. We will never discard it until it is discharged. When a second obligation appears later, say $\theta_2$, we maintain $\{ \theta_1, \theta_2 \}$. Now, when a third obligation appears, say $\theta_3$, we replace $\theta_2$ with $\theta_3$ to obtain $\{\theta_1, \theta_3\}$. Therefore, until $\theta_1$ gets discharged, the second obligation may keep getting replaced. Suppose $\theta_1$ gets discharged when the obligation-set is $\{\theta_1, \theta_k\}$. Now $\theta_k$ takes the role of $\theta_1$, and so on. Using this strategy, we can enforce every $\psi$-obligation generated is ultimately asserted.

The above idea can naturally be implemented by storing up to two obligations at each point: the earliest generated, and the latest generated, which can be determined by looking at the ages. In other words, if there are three or more obligations containing the same $\U_I$ formula with $I$ unbounded, we store two obligations, one with the largest value of the age, and the other with the smallest value of the age (since our obligations are well-formed, we would never have two obligations of the same type having the same age) and remove the rest. 

\subparagraph*{Dual Until.}
For Dual until, the optimization merges two obligations into one. We start with the case when $I$ is a bounded interval, with left end-point $l$ and right end-point $u$. From the semantics, a dual Until obligation $\theta_1 = (\varphi_1 \dual{\U}_I \varphi_2, x_1, t_1)$ refers to an interval with end-points $l - x_1$ and $u - x_1$. When a fresh obligation $(\varphi_1 \dual{\U}_I \varphi_2, 0, u)$ appears, we check whether its corresponding interval with endpoints $l$ and $u$ overlap with the former: this will happen when either $l < u - x_1$, or when $l = u - x_1$ and $I$ is either left-closed or right-closed (see figure below): 
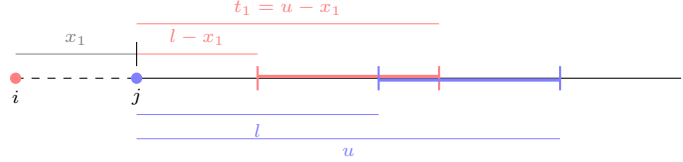
\begin{center}  
\begin{tikzpicture}[scale=0.8]
\draw (0,0) -- (9,0);
\draw [dashed] (-2, 0) to (0,0);

\node[circle, fill=red!50, inner sep=1.5pt] at (-2, 0) {};
\node[circle, fill=blue!50, inner sep=1.5pt] at (0, 0) {};

\node at (-2, -0.3) {\scriptsize $i$};
\node at (0, -0.3) {\scriptsize $j$};

\draw[thick, red!50] (2, 0.2) -- (2, -0.2);
\draw[thick, red!50] (5, 0.2) -- (5, -0.2);
\draw[very thin, red!50] (0, 0.4) to node [above] {\scriptsize $l -x_1$} (2, 0.4);
\draw[very thin, red!50] (0, 0.9) to node [above] {\scriptsize $t_1 = u - x_1$} (5, 0.9);
\draw[very thick,red!50] (2,0.03) -- (5,0.03);

\draw[thick, blue!50] (4, 0.2) -- (4, -0.2);
\draw[thick, blue!50] (7, 0.2) -- (7, -0.2);
\draw[very thin, blue!50] (0, -0.6) to node [below] {\scriptsize $l$} (4, -0.6);
\draw[very thin, blue!50] (0, -1) to node [below] {\scriptsize $u$} (7, -1);
\draw[very thick, blue!50] (4, -0.03) -- (7,-0.03);
\draw[very thin, gray] (-2, 0.4) to node [above] {\scriptsize $x_1$} (0, 0.4);
\draw[thin] (0,0.6) -- (0,0.2);

\end{tikzpicture}
\end{center}
More succinctly put, here is the definition for the function:
\begin{description}
\item[\mmerge] if $t_1 \in I$, remove $\theta_1, \theta_2$ and add a merged obligation $\theta = (\varphi_1 \U_I \varphi_2, x_1, u)$
\end{description}

Essentially, the life of the merged obligation is prolonged to cover the intervals represented by both $\theta_1$ and $\theta_2$. This explains the correctness of the merge. For the bound, let us look at two obligations $\theta_1 = (\varphi_1 \dual{\U}_I \varphi_2,x_1,t_1)$ and $\theta_2 = (\varphi_1 \dual{\U}_I \varphi_2,0,u)$ (a freshly created obligation).
When both $\theta_1$ and $\theta_2$ remain even after applying the \mmerge\ rule, it means that $t_1 < I$, and since it was started at $u$, it implies that $x_1$ is at least $|I|$. Hence, the time between these two obligations is at least $|I|$. Moreover, since $t_1 < I$, this obligation is alive only upto $l$ time units, and the number of new obligations that can get generated when it is alive is at most $1 + \lceil \frac{l}{|I|} \rceil$.

Finally, we consider the case when $I$ is unbounded. Here the interval corresponding to every fresh obligation will overlap with the previous one. Hence we can always merge: or in other words, any fresh obligation can be removed. Therefore, there will be at most one dual until obligation stored for formulas with an unbounded interval. 

\subparagraph*{The $\optimize$ function and the reduced obligation graph.}
We formalize the rules mentioned above into an $\optimize$ function. 
Given a node $(\Theta, \id)$ of the obligation graph $\ob(\varphi)$, we define a function $\optimize(\Theta, \id)$ which returns a new node $(\Thetaopt, \idopt)$ with potentially fewer obligations. Optimization is performed for each type independently: for each $\psi \in \type(\Theta)$, the $\optimize$ function restricts $\Theta$ to obligations of type $\psi$ and applies the rules corresponding to $\psi$ (whether $\U_I, \dual{\U}_I$, $I$ bounded or unbounded) to this set.  Using $\optimize$ we are able to define a \emph{reduced obligation graph} which is shown to be correct, and also containing a bounded number of obligations (details in Appendix~\ref{sec:append-red}). 

\begin{definition}[Reduced obligation graph]
     For an MTL formula $\varphi$, we define a \emph{reduced obligation graph} $\obred(\varphi)$. Each node of $\obred(\varphi)$ is a pair $(\Theta, \id)$ as in $\ob(\varphi)$. There is a special initial node $\init_\varphi$. The initial transitions are the same as in $\ob(\varphi)$. The intermediate transitions are defined as follows. For $\delta \ge 0$ and $a \in \Sigma$ we add $(\Theta, \id) \xRightarrow{\delta, a} (\Theta', \id')$ if $(\Theta, \id) \xra{\delta, a} (\overline{\Theta}, \overline{\id})$ in $\ob(\varphi)$ and $(\Theta', \id') = \optimize(\overline{\Theta}, \overline{\id})$. The acceptance condition is defined in the same way as in $\ob(\varphi)$. 
\end{definition}

We can now state the two main theorems of this section: the first one says that the optimization rules are correct for any MTL formula, the second one states that for MITL this results in a graph where there are a bounded number of obligations in each node.

\begin{restatable}{theorem}{oblredcorr}\label{thm:oblredcorr}
  Let $\varphi$ be an MTL formula and $w$  be a timed word. Then $w$ has an accepting run in $\obred(\varphi)$ iff $w$ has an accepting run in $\ob(\varphi)$.
\end{restatable}

\begin{restatable}{theorem}{mitlbdd}\label{thm:red-bdd-mitl}
Let $\varphi$ be an MITL formula. Then, each node in $\obred(\varphi)$ has at most $k_{\psi}$ obligations of type $\psi$, where (1) $k_{\psi} = 2 + 2 \times \lceil \frac{l}{u-l} \rceil$ for a bounded until subformula $\psi$ with endpoints $l,u$, (2) $k_\psi = 2$ for an unbounded until subformula $\psi$, (3) $k_\psi = 1 + \lceil \frac{l}{u-l} \rceil$ for a bounded dual until subformula $\psi$ with endpoints $l,u$, and (4) $k_\psi = 1$ for an unbounded dual until subformula $\psi$.
\end{restatable}

\section{Symbolic computation}
\label{sec:symbolic}

The goal of this section is to develop a finite abstraction of the reduced obligation graph $\obred(\varphi)$ corresponding to an MITL formula $\varphi$. This is possible because of the global bound on the number of obligations in a node that we obtained in the previous section. The idea is to adapt the standard region equivalence studied for timed automata~\cite{alur1994theory} in this context of obligations. The region equivalence is known to be a (time-abstract) bisimulation\footnote{We recall that a time-abstract bisimulation is an equivalence relation on configurations, where delays are matched by possibly different delays while discrete labelled transitions are matched by similarly-labelled transitions.} on the space of configurations of a timed automaton, and its finite quotient called the region graph is a sound and complete abstraction of the timed automaton semantics. 

The main challenge in our context lies in handling timer variables which move in an opposite direction to clocks, and it is known that for generalized timed automata (GTA) which contain both kinds of clocks (standard clocks and timers), there is in general no finite (time-abstract) bisimulation~\cite{DBLP:conf/cav/AkshayGGJS23}.
Here is a brief idea why this is so. The region equivalence is parameterized by a constant $M \in \N$. For a clock, once it goes beyond $M$, its actual value does not matter and can be abstracted into one chunk. However, for a timer, consider two different values $M + 1$ and $M + 2$. On a time elapse of $1$, the former comes to $M$, whereas the latter is still larger than $M$. Hence equating $M+1$ and $M+2$ is incorrect for timers. In our case, this does not happen since our timer variables are always bounded by design and are never above the maximum constant. 
We develop an equivalence relation which intuitively corresponds to applying the standard region equivalence when timer values are replaced by their negation. Note nevertheless that we will keep the semantics of timers as we defined (non-negative guessed value and decrease).

Fix an MITL formula $\varphi$. Let $M$ be the maximum constant appearing in the formula. For a timer $t$ (assuming non-negative values), let $\integraltimer{t}$ be the smallest integer greater than or equal to $t$, and let $\fractional{t} = \integraltimer{t} - t$. For a clock $x$, let $\integralclock{x}$ be the largest integer lesser than or equal to $x$, and let $\fractional{x} = x - \integralclock{x}$.

\begin{definition}[Region equivalence on obligation-sets]
Let $\Theta, \Theta'$ be well-formed obligation-sets. We say $\Theta \regeq_M \Theta'$ if there exists a bijection $h: \Theta \to \Theta'$ such that:
\begin{itemize}
\item for all $\theta \in \Theta$, both $\theta$ and $h(\theta)$ are of the same type; when $\theta = (\psi, x, t)$ and $h(\theta) = (\psi, x', t')$ we denote $h(x) = x'$ and $h(t) = t'$;
\item for every $(\psi, x, t) \in \Theta$ 
\begin{itemize}
\item $t = \bot$ iff $h(t) = \bot$; if $t \neq \bot$, then $\integraltimer{t} = \integraltimer{h(t)} $ and $\fractional{t} = 0$ iff $\fractional{h(t)} = 0$
\item $x > M$ iff $h(x) > M$; when $x \le M$, $\integralclock{x} = \integralclock{h(x)}$ and $\fractional{x} = 0$ iff $\fractional{h(x)} = 0$ 
\end{itemize} 
\item for two values $y, z$ (clock or timer) appearing in $\Theta$, we have $\fractional{y} \le \fractional{z}$ iff $\fractional{h(y)} \le \fractional{h(z)}$.
\end{itemize}  
\end{definition} 

We can extend the equivalence to nodes: $(\Theta, \id) \regeq_M (\Theta', \id')$ if $\Theta \regeq_M \Theta'$ by a bijection $h$, and for any $\theta \in \Theta$, $\id(\theta) = \id'(h(\theta))$. We will write $[(\Theta, \id)]_M$ for the equivalence class of $(\Theta, \id)$ w.r.t. $\regeq_M$. 

\begin{restatable}{lemma}{regionbisimulation}\label{lem:region-correct}
Let $\varphi$ be an MITL formula. The relation $\regeq_M$ on the nodes of $\obred(\varphi)$ is a time-abstract bisimulation. Furthermore, the equivalence relation induces a finite index.
\end{restatable}
This leads us to the definition of a region graph over obligations. We will call the equivalence classes as obligation-regions and its quotient as the obligation-region graph. 

\begin{definition}[Region graph of obligations for MITL formulas.]
Let $\varphi$ be an MITL formula. Nodes of the \emph{obligation-region graph} are the equivalence classes of $\regeq_M$ over nodes of $\obred(\varphi)$, along with a special node $\init_\varphi$ to denote the initial node. 
For every $a \in \Sigma$, there is a transition $\init_\varphi \xrightarrow{a} [(\Theta_0, \id_0)]_M$ if there exists a transition $\init_\varphi \xrightarrow{\delta, a} (\Theta_0, \id_0)$ for some $\delta \in \Rpos$ in $\obred(\varphi)$. Similarly, 
there is a transition $[(\Theta, \id)]_M \xrightarrow{a} [(\Theta', \id')]_M$ if there exists $\delta \in \Rpos$ s.t. $(\Theta, \id) \xRightarrow{\delta, a} (\Theta', \id')$. 
\end{definition}

The acceptance condition is adapted from $\obred(\varphi)$. For every bounded Until subformula $\psi$ of $\varphi$, the number of obligations containing $\psi$ in any obligation-region is atmost $k_\psi$ (the bound of Theorem~\ref{thm:red-bdd-mitl}). For $1 \le i \le k_\psi$, let $G_{\psi}^i$ be the set $[(\Theta, \id)]_M$ s.t. $(\Theta, \id) \in F_{\psi}^i$ (coming from Definition~\ref{def:acc-condition}). An infinite run of the obligation region-graph is said to be accepting if for every bounded Until formula $\psi$ and every $1 \le i \le k_\psi$, some node in $G_\psi^i$ appears infinitely often. This is a standard generalized B\"uchi condition. 

\begin{restatable}{theorem}{oblregcorr}\label{thm:obl-region-correct}
Let $\varphi$ be an MITL formula. The obligation-region graph is sound and complete: there exists an accepting run in $\obred(\varphi)$ iff there is an accepting run in the obligation-region graph of $\varphi$. 
\end{restatable}

Recall that satisfiability of $\varphi$ is the question of checking whether there exists a non-Zeno word satisfying $\varphi$. We can encode non-Zenoness using an MITL formula $\phi_{nz} := \G (\top \implies \F_{[1, \infty)} \top)$. Hence satisfiability of $\varphi$ reduces to checking whether there is an accepting run in the obligation-region graph of $\varphi \land \phi_{nz}$. 
Thanks to Theorem~\ref{thm:red-bdd-mitl}, each node of the obligation-region graph requires at most exponential space and the size of the obligation-region graph is at most doubly exponential in the formula, when constants are encoded in binary. When constants are encoded in unary, each node requires polynomial space and the graph is at most singly-exponential in the formula.  
Hence, we get the following corollary.

\begin{restatable}{corollary}{mitlcomp}
  The satisfiability problem for MITL can be solved in EXPSPACE, and if constants are encoded in unary, it can be solved in PSPACE. For the fragment MITL$_{0,\infty}$ that uses only $0$ and $\infty$ as constants, the problem can also be solved in PSPACE.
\end{restatable}
Note that the EXPSPACE upper bound coincides with the result of~\cite{alur1996benefits}, as well as the PSPACE upper bound for MITL$_{0,\infty}$. We did not find the other PSPACE upper bound mentioned in the previous literature. These complexity results are optimal, in the sense that the satisfiability problem for MITL is known to be EXPSPACE-hard, and the other logics encompass LTL, so PSPACE is optimal.

\section{Conclusion}
\label{sec:conclusion}

We have presented an alternate approach to MITL satisfiability. In spite of its long history, there has been a recent flurry of activity in MITL satisfiability as witnessed by two new orthogonal approaches leading to mature tools~\cite{DBLP:conf/tacas/HoKMMP26,DBLP:conf/tacas/AkshayCGGS26}. We believe our approach complements the existing approaches and we hope that some of the insights we develop here could benefit the other approaches too. Both the existing tools for MITL started with a ``seed'' laying the foundations: \cite{brihaye2013mitl,DBLP:conf/formats/BrihayeEG14} for the tool MightyL and its successor MightyPPL, and \cite{DBLP:conf/concur/0001G0S24} for TEMPORA. In our perspective, we view the current work as the seed for further investigations into this approach. Therefore, a natural future direction is the study of zone-based algorithms and tool development based on our simple obligations for MITL. 

\bibliography{ATA-zones}

\appendix 
\section{Appendix for Section \ref{sec:obligations}}\label{sec:append-ob}

\wellformed*
\begin{proof}
We take some $\delta \in \Rpos$, $a \in \Sigma$ and some arbitrary $\Theta'$ such that $\Theta'$ is a minimal $(\delta,a)$-extension of $\Theta$, and prove that $\Theta'$ is also well-formed by contradiction. Assume that $\Theta'$ is not well-formed. This means there exist obligations $(\psi,x,t_1),(\psi,x,t_2)$ for some subformula $\psi = \psi_1 B_I \psi_2$ such that $t_1 \neq t_2$. Now, there are two possible cases:
\begin{itemize}
    \item[-]If $x > 0$, it means that $(\psi,x-\delta,t_1 + \delta)$ and $(\psi,x-\delta,t_2+\delta)$ would also be present in $\Theta$, by definition of transitions. This is a contradiction as $\Theta$ is well-formed.
    \item[-]If $x = 0$, it means $(\psi,0,t_1)$ and $(\psi,0,t_2)$ are newly created obligations, and thus $t_1,t_2 \in I$. This in turn means that $\Theta''_{i} = \Theta' \setminus (\{(\psi,0,t_1),(\psi,0,t_2)\}) \cup (\psi,0,t_i)$ for $i = \{1,2\}$ are both valid $(\delta,a)$-extensions of $\Theta$. But $\Theta''_{i} \subset \Theta'$, which is a contradiction as $\Theta''$ is a minimal $(\delta,a)$-extension of $\Theta$.  
\end{itemize}
Hence, $\Theta'$ will also be well formed for any $\Theta'$ that is a minimal $(\delta,a)$-extension of a well-formed obligation-set $\Theta$.
\end{proof}

\begin{figure}
    \centering
    \begin{tikzpicture}
    \begin{scope}[scale=0.8]
        \begin{scope}[rounded corners]
            \fill[gray!50] (5.8,5.5) rectangle (6.2,4.9);
            \fill[gray!50] (3.8,4.5) rectangle (4.2,3.9);
            \fill[violet!30] (5.7,4.5) rectangle (6.3,3.9);
            \fill[blue!50] (5.2,3.5) rectangle (6.8,2.9);
            \fill[blue!20] (3,2.5) rectangle (4,1.9);
            \fill[blue!50] (5.2,2.5) rectangle (6.8,1.9);
            \fill[red!50] (7.8,2.5) rectangle (10.2,1.9);
            \fill[gray!50] (-0.2,1) rectangle (0.2,0.4);
            \fill[blue!20] (0.5,1) rectangle (1.5,0.4);
            \fill[blue!20] (1.7,1) rectangle (2.7,0.4);
            \fill[blue!50] (3,1) rectangle (4.6,0.4);
            \fill[red!50] (4.8,1) rectangle (7.2,0.4);
            \fill[red!20] (7.4,1) rectangle (9,0.4);
            \fill[red!50] (9.4,1) rectangle (11.6,0.4);
            \fill[orange!50] (11.8,1) rectangle (15.2,0.4);
            \foreach \i in {0.4, 1.9}{
                \fill[gray!50] (\i-0.2,-0.2) rectangle (\i+0.2,-0.8);
            }
            \foreach \i in {1, 2.5, 3.3, 7.9}{
                \fill[blue!20] (\i-0.2,-0.2) rectangle (\i+0.2,-0.8);
            }
            \foreach \i in {3.9}{
                \fill[blue!50] (\i-0.2,-0.2) rectangle (\i+0.2,-0.8);
            }
            \foreach \i in {4.6, 6.1, 10.6}{
                \fill[red!50] (\i-0.2,-0.2) rectangle (\i+0.2,-0.8);
            }
            \foreach \i in {5.5, 8.5, 10}{
                \fill[red!20] (\i-0.2,-0.2) rectangle (\i+0.2,-0.8);
            }
            \foreach \i in {6.7, 11.4, 13.5}{
                \fill[orange!50] (\i-0.2,-0.2) rectangle (\i+0.2,-0.8);
            }
            \foreach \i in {12.6}{
                \fill[orange!20] (\i-0.2,-0.2) rectangle (\i+0.2,-0.8);
            }
            \foreach \i in {14.4}{
                \fill[green!50] (\i-0.2,-0.2) rectangle (\i+0.2,-0.8);
            }
        \end{scope}
        \node (-1) at (6,5.2) {\scriptsize $\{\}$};
        \node (01) at (4,4.2) {\scriptsize $\{\}$};
        \node (02) at (6,4.2) {\scriptsize $(\varphi)$};
        \node (12) at (6,3.2) {\scriptsize $(\varphi),(\varphi_1)$}; 
        \node (21) at (3.5,2.2) {\scriptsize $(\varphi_1)$};
        \node (22) at (6,2.2) {\scriptsize $(\varphi),(\varphi_1)$};
        \node (23) at (9,2.2) {\scriptsize $(\varphi),(\varphi_1),(\varphi_1)$};
        \node (31) at (0,0.7) {\scriptsize $\{\}$};
        \node (32) at (1,0.7) {\scriptsize $(\varphi_1)$};
        \node (33) at (2.2,0.7) {\scriptsize $(\varphi_1)$};
        \node (34) at (3.8,0.7) {\scriptsize $(\varphi),(\varphi_1)$};
        \node (35) at (6,0.7) {\scriptsize $(\varphi),(\varphi_1),(\varphi_1)$};
        \node (36) at (8.2,0.7) {\scriptsize $(\varphi_1),(\varphi_1)$};
        \node (37) at (10.5,0.7) {\scriptsize $(\varphi),(\varphi_1),(\varphi_1)$};
        \node (38) at (13.5,0.7) {\scriptsize $(\varphi),(\varphi_1),(\varphi_1),(\varphi_1)$};
        \node (d) at (6,-1.2) {\scriptsize $\dots$};
        \begin{scope}[black!80,->,>=stealth]
            \draw (-1) to (01);
            \draw (-1) to (02);
            \draw (02) to (12);
            \draw (12) to (21);
            \draw (12) to (22);
            \draw (12) to (23);
            \draw (21) to (31);
            \draw (21) to (32);
            \draw (22) to (33);
            \draw (22) to (34);
            \draw (22) to (35);
            \draw (23) to (36);
            \draw (23) to (37);
            \draw (23) to (38);
            
            \draw (32) to (0.4,-0.1);
            \draw (32) to (1,-0.1);
            \draw (33) to (1.9,-0.1);
            \draw (33) to (2.5,-0.1);
            \draw (34) to (3.4,-0.1);
            \draw (34) to (3.9,-0.1);
            \draw (34) to (4.6,-0.1);
            \draw (35) to (5.5,-0.1);
            \draw (35) to (6.1,-0.1);
            \draw (35) to (6.7,-0.1);
            \draw (36) to (7.9,-0.1);
            \draw (36) to (8.5,-0.1);
            \draw (37) to (10,-0.1);
            \draw (37) to (10.6,-0.1);
            \draw (37) to (11.4,-0.1);
            \draw (38) to (12.6,-0.1);
            \draw (38) to (13.5,-0.1);
            \draw (38) to (14.4,-0.1);
        \end{scope}
    \end{scope}

    \end{tikzpicture}
    \caption{An abstract representation of $\ob(\varphi)$}\label{fig:transition-system}
\end{figure}
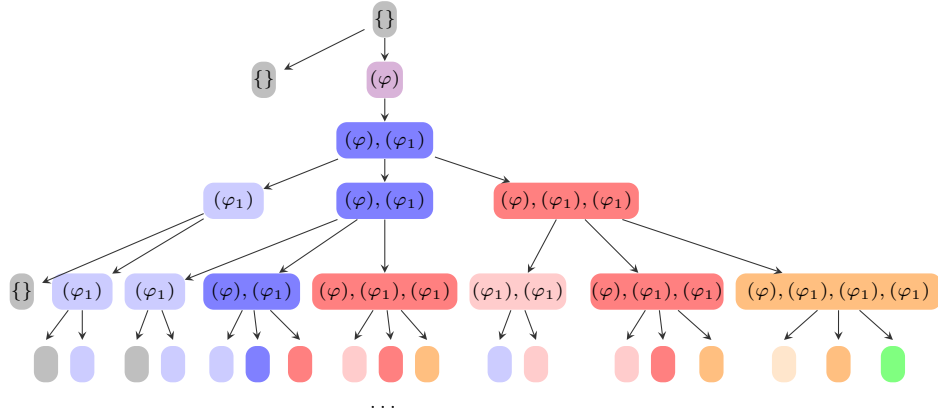
\begin{example}
    Figure \ref{fig:transition-system} gives an abstract picture of the obligation graph generated for the until formula $\varphi = \varphi_1 \U_{[2,10]} b$, where $\varphi_1 = (\F_{[1,3]} a)$. The figure shows the type of obligations generated for $\varphi$, and gives a rough idea of how $\ob(\varphi)$ looks like. 
\end{example}

Before turning to the proof of Theorem~\ref{thm:oblcorr}, we state three useful lemmas whose proofs follow from the definitions. We let $\varphi$ be an MTL formula, and consider its obligation graph $\ob(\varphi)$. In these lemmas, we use all notations that were previously defined.

\begin{lemma}
  \label{lemma:ids}
  Assume $(\Theta_0,\id_0) \xrightarrow{\delta_1,a_1} (\Theta_1,\id_1) \xrightarrow{\delta_2,a_2} \dots$ is an accepting run in $\ob(\varphi)$. Then for every $\id'_0$ that is a naming function for $\Theta_0$, the run $(\Theta_0,\id'_0) \xrightarrow{\delta_1,a_1} (\Theta_1,\id'_1) \xrightarrow{\delta_2,a_2} \dots$ is an accepting run.  \end{lemma}

\begin{lemma}
  \label{lemma:proj}
  Assume $(\Theta_0,\id_0) \xrightarrow{\delta_1,a_1} (\Theta_1,\id_1) \xrightarrow{\delta_2,a_2} \dots$ is an accepting run in $\ob(\varphi)$. Then for every $\Theta'_0 \subseteq \Theta_0$ and $\id'_0$ which restricts $\id_0$ to $\Theta'_0$, there is a unique accepting run $(\Theta'_0,\id'_0) \xrightarrow{\delta_1,a_1} (\Theta'_1,\id'_1) \xrightarrow{\delta_2,a_2} \dots$  in $\ob(\varphi)$, where $\Theta'_j \subseteq \Theta_j$ for every $j$. It makes the same (non-deterministic) choices along the run, and it is called the \emph{projection} of the original run on obligation-set $\Theta'_0$.
\end{lemma}

\begin{lemma}
  \label{lemma:union}
  Assume $(\Theta_0,\id_0) \xrightarrow{\delta_1,a_1} (\Theta_1,\id_1) \xrightarrow{\delta_2,a_2} \dots$  in $\ob(\varphi)$ and $(\Theta'_0,\id'_0) \xrightarrow{\delta_1,a_1} (\Theta'_1,\id'_1) \xrightarrow{\delta_2,a_2} \dots$ are accepting runs in $\ob(\varphi)$, and assume that the ranges of $\id_0$ and $\id'_0$ are disjoint (which we can assume wlog due to Lemma~\ref{lemma:ids}). Then for every maximally well-formed subset $\Theta''_0$ of $\Theta_0 \cup \Theta'_0$, there is an accepting run $(\Theta''_0,\id''_0) \xrightarrow{\delta_1,a_1} (\Theta''_1,\id''_1) \xrightarrow{\delta_2,a_2} \dots$ such that $\Theta''_j \subseteq \Theta_j \cup \Theta'_j$ for every $j$, and $\id''_0$ restricts $\id_0 \cup \id'_0$ to $\Theta''_0$. We call it a \emph{join} of the two initial accepting runs.
\end{lemma}
Basically, a join of two runs is a union of runs, where some obligations are removed to only have well-formed obligation-sets. Identifiers need to be adapted. 
\oblcorr*

\begin{proof}
  Assume $w = (a_1,\tau_1) (a_2,\tau_2) \dots$
    To prove this, we first make the claim that for any subformula $\psi$ of $\varphi$, $w,i \sat \psi$ iff there is an accepting run of the form $(\Theta,\id) \xra{\delta_{i+1},a_{i+1}} (\Theta_{i+1},\id_{i+1}) \dots$ in $\ob(\varphi)$  (where $\delta_j = \tau_j-\tau_{j-1}$ for every $j$) for some $\Theta \in \Ii_{\psi,a_{i}}$ and any identifier $\id$ for $\Theta$ (the initial choice of $\id$ does not matter, acceptance is independent of it).
    Using this claim, it is straightforward to see that the theorem holds: there is an accepting run for $w$ in $\ob(\varphi)$ iff there is an accepting run for $(a_{2},\tau_2) (a_3,\tau_3) \dots$ from $(\Theta_1,\id_1)$ for some $\Theta_1 \in \Ii_{\varphi,a_{1}}$ and identifier $\id_1$ for $\Theta_1$, and this holds, using the claim, iff $w,1 \sat \varphi$.
    We now prove this claim using an induction on the structure of the subformula $\psi$. 
    \begin{description}
        \item[Base cases] When $\psi = a$, $w,i \sat \psi$ if $a_{i} = a$. As $\Ii_{\psi,a_{i}}$ is defined only if $a_{i} = a$, in which case there is a run with no obligations generated ($\Ii_{\psi,a_{i}} = \Theta_\emptyset$), hence an accepting run from $(\Theta_\emptyset,\id_\emptyset)$, and so the statement holds. We can similarly prove the statement for the case when $\psi = \lnot a$.
        \item[I.H.] We assume that the claim holds for subformulas of $\psi$, and look at the possible cases:
          \begin{itemize}
        \item for $\psi = \psi_1 \lor \psi_2$ and $\psi = \psi_1 \land \psi_2$, the proof is straightforward because of the way we have defined the initial obligation-sets, and thanks to the technical Lemmas~\ref{lemma:union} and~\ref{lemma:proj}. 
            \item if $\psi = \psi_1 \U_I \psi_2$ and $I$ is bounded, then $(w,i) \sat \psi$ if there is some $i < j$ such that $(w,j) \sat \psi_{2}$, $\tau_j - \tau_i \in I$ and for all $i < k < j$, $(w,k) \sat \psi_1$. Now, as $I$ is bounded, $\Theta \in \Ii_{\psi,a_{i}}$ will be of the form $\{(\psi , 0, t)\}$ for $t \in I$. There is an accepting run from $(\Theta,\id)$ on $(a_{i+1},\tau_{i+1})\dots$ only if this initially generated obligation is removed in a later obligation-set, along with every other generated until obligation during the run. The obligation in $\Theta$ is removed at some transition only if there is an event that occurs at $t$ time units from $i$, where obligations for $\psi_2$ will be generated. Also, at every transition until that point, additional obligations for $\psi_1$ will be generated. When new obligations are generated, accepting runs are merged by joining runs (Lemma~\ref{lemma:union}). Assuming that the induction hypothesis for $\psi_1,\psi_2$, this means that there is an accepting run from $\Theta$ if there is a point $i < j$ such that $\tau_j - \tau_i = t$ (where $t \in I$), $w,j \sat \psi_2$, and for every point $i < k < j$, $w,k \sat \psi_1$. Hence the statement holds for this case.
            \item if $ \psi = \psi_1 \U_I \psi_2$ and $I$ is unbounded, then $w,i \sat \psi$ if there is some $i < j$ such that $(w,j) \sat \psi_{2}$, $\tau_j - \tau_i \in I$ and for all $i < k < j$, $(w,k) \sat \psi_1$. Now, as $I$ is an unbounded interval, $\Theta$ will be exactly $\{(\psi,0,\bot)\}$, and this obligation is removed, similar to before, if there is some event that occurs after time units $x \in I$ from $i$, where obligations for $\psi_2$ are generated, and every event until then generates obligation for $\psi_1$ additionally. The reasoning then will follow similarly.
            
            \item if $\psi = \psi_1 \dual{\U} \psi_2$ and $I$ is bounded, $w,i \sat \psi$ if for all $j$ s.t. $i < j$ and $\tau_j - \tau_i \in I$, either $(w,j) \sat \psi_{2}$, or there is some $k$ with $i < k < j$ with $(w,k) \sat \psi_{1}$. Here, $\Theta$ will be $\{(\psi,0,u)\}$, and as the obligation present is not an until obligation, we only need to ensure that the additional until obligations generated in the course of the run are eventually removed. Notice that once an event occurs after $c$ time units from $i$, this obligation will get removed, and so we only need to worry about the obligations generated by transitions up until such an event occurs, or this obligation gets removed. Now, during the run on $(a_{i+1},\tau_{i+1})\dots$, when the time elapsed is less than the interval $I$, either no new obligation is generated, or an obligation for $\psi_1$ is generated, and the original obligation is removed. When the time elapsed is within the interval, transitions generate additional obligations for $\psi_2$, or obligations for $\psi_1$ and $\psi_2$ and remove the original obligation. Using the induction hypothesis, this means that there is an accepting run from $\Theta$ if (1) there is some position $i<j$ such that $w,j \sat \psi_{1}$ and $\tau_{j} - \tau_{i} < I$, or (2) there is some point $i < k'$ such that $w,k' \sat \psi_1$, $w,k' \sat \psi_2$, and $\tau_{k'} - \tau_i \in I$, or (3) for every point $i < k$ such that $\tau_{k} - \tau_i \in I$, $w,k \sat \psi_2$. To see why this is equivalent to saying $w,i \sat \psi$, we see that if there is some position $l$ with $\tau_l - \tau_i \in I$ such that $w,l \not\sat \psi_2$, its corresponding previous position $l'$ such that $w,l' \sat \psi_1$ either occurs before the interval $I$(checked in case 1), or during the interval $I$, where either $w,l' \sat \psi_2$ also, or we need to again check some previous position for this position that will satisfy $\psi_1$, meaning eventually, we will reach a position before the interval that satisfies $\psi_1$, or a position within the interval that satisfies both $\psi_1$ and $\psi_2$. Hence statement holds.
            
            \item if $\psi = \psi_1 \dual{\U} \psi_2$ and $I$ is an unbounded interval, $w,i \sat \psi$ if for all $j$ s.t. $i < j$ and $\tau_j - \tau_i \in I$, either $(w,j) \sat \psi_{2}$, or there is some $k$ with $i < k < j$ with $(w,k) \sat \psi_{1}$. Here, $\Theta$ will be $\{(\psi,0,\bot)\}$, and as before we need to ensure that the additional until obligations generated in the course of the run are eventually removed. The argument follows similarly to the previous case here, the only difference being that the obligation can only be removed due to a discrete transition rather than a time elapse. Now, in the run, when the time elapsed is less than the interval $I$, either no new obligation is generated, or an obligation for $\psi_1$ is generated, and the original obligation is removed. When the time elapsed is within the interval, transitions generate additional obligations for $\psi_2$, or obligations for $\psi_1$ and $\psi_2$ and remove the original obligation. Using the induction hypothesis, this means that there is an accepting run from $\Theta$ if (1) there is some position $i<j$ such that $w,j \sat \psi_{1}$ and $\tau_{j} - \tau_{i} < I$, or (2) there is some point $i < k'$ such that $w,k' \sat \psi_1$, $w,k' \sat \psi_2$, and $\tau_{k'} - \tau_i \in I$, or (3) for every point $i < k$ such that $\tau_{k} - \tau_i \in I$, $w,k \sat \psi_2$. Hence, as before, the statement holds.
        \end{itemize} 
    \end{description}
\end{proof}

\section{Appendix for Section \ref{sec:removing-redundant-obligations}}\label{sec:append-red}
\subsection{Reduced Obligation Graph}

Consider an obligation graph $\ob(\varphi)$ for an MTL formula $\varphi$. 
Given a node $(\Theta, \id)$ we define a function $\optimize(\Theta, \id)$ which returns a new node $(\Thetaopt, \idopt)$ with potentially fewer obligations than $\Theta$, and a possibly changed naming function. For each subformula $\psi$ of $\varphi$ containing an $\U_I$ or a $\dual{\U}_I$ operator, assume there are $\ell_\psi \ge 0$ obligations of type $\psi$ in $\Theta$. The number $\ell_\psi$ can well be $0$. When $\ell_\psi > 0$, let the corresponding obligations be called: $\theta^\psi_i = (\psi, x^\psi_i, t^\psi_i)$ for $1 \le i \le \ell_{\psi}$. Since $\Theta$ is well-formed, we can assume that the ages are distinct, and furthermore, for convenience, assume: $x^\psi_{\ell_\psi} < x^{\psi}_{\ell_\psi - 1} < \cdots < x^\psi_2 < x^\psi_1$: the obligations were produced in increasing order of their indices, and hence the one with the largest index has the smallest age. The reduction rules depend on the type of $\psi$. We make use of this notation: for an obligation-set $\Theta$ and a formula $\psi$, we write $\Theta^\psi$ for the set of obligations of type $\psi$ present in $\Theta$. 

\begin{definition}[Until with bounded interval]\label{def:until_bdd_opt}
Let $\psi = \varphi_1 \U_I \varphi_2$ where $I$ is a bounded interval with $l, u$ being the left and right end-points respectively. When $x^\psi_{\ell_\psi} = 0$, apply the following reduction rules. 
\begin{description}
\item[\removenew] If $t^\psi_{\ell_\psi-1} \in I$, then remove $\theta_{\ell_\psi}$, i.e., $\Thetaopt^\psi = \Theta^\psi \setminus \theta_{\ell_\psi}$, and $\idopt^\psi = \id^\psi$. 
\item[\modifyold] Else, if $x^\psi_{\ell_\psi-1} + t^\psi_{\ell_\psi} \in I$, then $\Thetaopt^\psi = \Theta^\psi \setminus \{\theta_{\ell_\psi - 1},\theta_{\ell_\psi}\} \cup \theta$ for $\theta = (\psi,x^\psi_{\ell_\psi - 1},t^\psi_{\ell_\psi})$; in this case, define $\idopt(\theta) = \id(\theta_{\ell_{\psi} - 1})$, and $\idopt(\theta_m) = \id(\theta_m)$ for all $1 \le m \le \ell_{\psi - 2}$. 
\item[\keepboth]If $t^\psi_{\ell_{\psi}-1} \not\in I$ and $x^\psi_{\ell_{\psi}-1} + t^\psi_{\ell_{\psi}} \not\in I$, then $\Thetaopt^{\psi} = \Theta^\psi$ and $\idopt^\psi = \id^\psi$.
\end{description}
\end{definition}

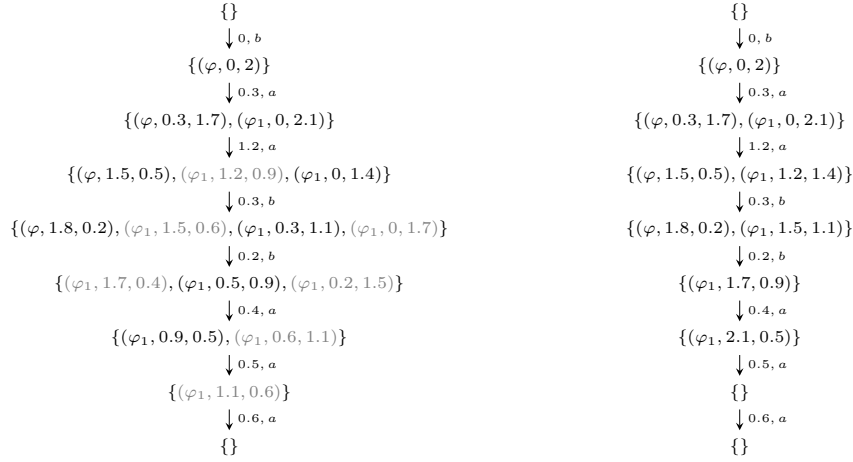
\begin{figure}
    \centering
    \scalebox{0.9}{
    \begin{tikzpicture}
    \begin{scope}[xshift=1.5cm,scale=0.8]
        \node (0) at (0, 3.5) {\scriptsize $\{\}$};
        \node (1) at (0, 2.5) {\scriptsize $\{(\varphi,0,2)\}$};        
        \node (2) at (0, 1.5) {\scriptsize $\{(\varphi,0.3,1.7),(\varphi_1,0,2.1)\}$};        
        \node (3) at (0, 0.5) {\scriptsize $\{ (\varphi, 1.5, 0.5),\textcolor{gray}{(\varphi_1, 1.2, 0.9)}, (\varphi_1, 0, 1.4)\}$}; 
        \node (4) at (0, -0.5) {\scriptsize $\{(\varphi, 1.8, 0.2), \textcolor{gray}{(\varphi_1, 1.5, 0.6)}, (\varphi_1, 0.3, 1.1), \textcolor{gray}{(\varphi_1, 0, 1.7)} \}$};
        \node (5) at (0, -1.5) {\scriptsize $\{\textcolor{gray}{(\varphi_1, 1.7, 0.4)}, (\varphi_1, 0.5, 0.9), \textcolor{gray}{(\varphi_1, 0.2, 1.5)} \}$};
        \node (6) at (0, -2.5) {\scriptsize $\{{(\varphi_1, 0.9, 0.5)}, \textcolor{gray}{(\varphi_1, 0.6, 1.1)}\}$};
        \node (7) at (0, -3.5) {\scriptsize $\{ \textcolor{gray}{(\varphi_1, 1.1, 0.6)}\}$};
        \node (8) at (0, -4.5) {\scriptsize $\{\}$};
    \end{scope}
    \begin{scope}[->, >=stealth]
        \draw (0) to node [right] {\tiny $0, b$} (1);
        \draw (1) to node [right] {\tiny $0.3, a$} (2);
        \draw (2) to node [right] {\tiny $1.2, a$} (3);
        \draw (3) to node [right] {\tiny $0.3, b$} (4);
        \draw (4) to node [right] {\tiny $0.2, b$} (5);
        \draw (5) to node [right] {\tiny $0.4, a$} (6);
        \draw (6) to node [right] {\tiny $0.5, a$} (7);
        \draw (7) to node [right] {\tiny $0.6, a$} (8);
    \end{scope}
    \begin{scope}[xshift=9cm,scale=0.8]
        \node (0) at (0, 3.5) {\scriptsize $\{\}$};
        \node (1) at (0, 2.5) {\scriptsize $\{(\varphi,0,2)\}$};        
        \node (2) at (0, 1.5) {\scriptsize $\{(\varphi,0.3,1.7),(\varphi_1,0,2.1)\}$};        
        \node (3) at (0, 0.5) {\scriptsize $\{ (\varphi, 1.5, 0.5), (\varphi_1, 1.2, 1.4)\}$}; 
        \node (4) at (0, -0.5) {\scriptsize $\{(\varphi, 1.8, 0.2), (\varphi_1, 1.5, 1.1)\}$};
        \node (5) at (0, -1.5) {\scriptsize $\{ (\varphi_1, 1.7, 0.9)\}$};
        \node (6) at (0, -2.5) {\scriptsize $\{{(\varphi_1, 2.1, 0.5)}\}$};
        \node (7) at (0, -3.5) {\scriptsize $\{\}$};
        \node (8) at (0, -4.5) {\scriptsize $\{\}$};
    \end{scope}
    \begin{scope}[->, >=stealth]
        \draw (0) to node [right] {\tiny $0, b$} (1);
        \draw (1) to node [right] {\tiny $0.3, a$} (2);
        \draw (2) to node [right] {\tiny $1.2, a$} (3);
        \draw (3) to node [right] {\tiny $0.3, b$} (4);
        \draw (4) to node [right] {\tiny $0.2, b$} (5);
        \draw (5) to node [right] {\tiny $0.4, a$} (6);
        \draw (6) to node [right] {\tiny $0.5, a$} (7);
        \draw (7) to node [right] {\tiny $0.6, a$} (8);
    \end{scope}

    \end{tikzpicture}
    }
    \caption{Eliminating obligations for bounded until}
    \label{fig:oblred-until}
\end{figure}

\begin{example}
    The Figure \ref{fig:oblred-until} depicts the eliminate rules being applied to bounded until formulas. The formula $\varphi$ and the word $w$ are same as the example in Figure \ref{fig:obligation-runs}. The right shows the run for the word on $\ob(\varphi)$ as before, with the obligations that will be eliminated in gray. The first such obligation-set when an elimination rule can be applied is the one reached after the third transition. Here, there is a newly generated obligation for $\varphi_1$ along with an existing old one for the same formula. We look at the values and see that we can apply the \modifyold\ case. In the next transition, we again have a newly generated obligation for $\varphi_1$, and we see that the \removenew\ case is applicable and hence the newly generated obligation can be removed. The left shows the final run in $\obred(\varphi)$ for this word.
\end{example}

\begin{definition}[Until with unbounded interval]\label{def:until_ubdd_opt}
Let $\psi = \varphi_1 \U_I \varphi_2$ with $I$ an unbounded interval with left endpoint $b$. Keep only the earliest and the latest witnesses: in other words, $\Thetaopt^\psi = \{ \theta_1, \theta_{\ell_\psi}\}$; $\idopt(\theta_1) = 1$, $\idopt(\theta_{\ell_\psi}) = 2$. 
\end{definition}

\begin{definition}[Dual until]\label{def:dual_until_opt}
Let $\psi = \varphi_1 \dual{\U}_I \varphi_2$. When $I$ is bounded, $x^\psi_{\ell_\psi} = 0$, and $\ell_\psi \ge 2$, apply the following rules to merge the last two obligations. 
\begin{description}
\item[\mmerge] if $t^\psi_{\ell_\psi-1} \in I$, then $\Thetaopt^\psi = \Theta^\psi \setminus \{\theta_{\ell_\psi - 1}, \theta_{\ell_\psi}\} \cup \{(\psi, x^\psi_{\ell_{\psi-1}}, t^\psi_{\ell_\psi})\}$. 
\item[\keepboth] if $t^\psi_{\ell_\psi-1} < I$, then $\Thetaopt^\psi = \Theta^\psi$. 
\end{description}
When $I$ is unbounded, simply keep the earliest obligation: $\Thetaopt^\psi = \{ \theta_1\}$. 
\end{definition}

\subsection{Proof of Theorem \ref{thm:oblredcorr}}
We will first prove that if $w = (a_1,\tau_1) (a_2,\tau_2) \dots$ has an accepting run in $\obred(\varphi)$, then it has an accepting run in $\ob(\varphi)$. 

Consider an accepting run for $w$ in $\obred(\varphi)$ of the following form, with $\delta_i = \tau_i - \tau_{i-1}$ (taking $\tau_0 = 0$):
\begin{align*}
\rho:= \init_\varphi \xRightarrow{(\delta_1, a_1)} (\Theta_1, \id_1) \xRightarrow{\delta_2, a_2} (\Theta_2, \id_2) \cdots
\end{align*}
Firstly, we know that for every transition $(\Theta_{j-1}, \id_{j-1}) \xRightarrow{\delta_j, a_j} (\Theta_j, \id_j)$ of $\rho$ there is a transition $(\Theta_{j-1}, \id_{j-1}) \xra{\delta_j, a_j} (\overline{\Theta}_{j}, \overline{\id}_j)$ in $\ob(\varphi)$ that was optimized to get the resulting node in the reduced graph.  We call this transition:

\begin{align}\label{eq:original-transition}
e_j: (\Theta_{j-1}, \id_{j-1}) \xra{\delta_j, a_j} (\overline{\Theta}_{j}, \overline{\id}_j) 
\end{align}

Another thing to observe is as follows:  
\begin{lemma}\label{lem:single-witness-in-interval}
For a subformula $\psi$ of $\varphi$ of the form $\psi_1 \U_I \psi_2$ where $I$ is bounded, and some node node $(\Theta_j, \id_j)$ in the run $\rho$, there is at most one obligation $(\psi, x, t) \in \Theta_j$ such that $t \in I$.
\end{lemma}
\begin{proof}
    We will prove this by contradiction, assuming there are two obligations $\theta = (\psi,x,t)$ and $\theta' = (\psi,x',t')$ such that $\theta,\theta' \in \Theta_j$ and $t,t' \in I$. We can assume $x > x'$ here without loss of generality. Now, as $\theta,\theta' \in \Theta_j$, it means there is some $\Theta_k$ in $\rho$ such that $k < j$, $\tau_j - \tau_k = x'$ and there are obligations $\theta_k = (\psi, x - x', t + x')$ and $\theta'_k = (\psi,0,t'+x')$ in $\Theta_k$. Because both obligations were kept in $\Theta_k$, it means that \removenew\ and \modifyold\ were not applicable, hence $t + x' \not\in I$ and $t' + x' + x - x' \not \in I$. As $x' \leq x$ and $t \in I$, this means that $x + t > I$. Now, going further back to position $l$ such that $l < k$, $\tau_k - \tau_l = x - x'$, we know there will be an obligation $\theta_l = (\psi,0,t+x'+x-x') = (\psi,0,t+x)$ present in $\Theta_l$. This means $x+t \in I$, which is a contradiction.
\end{proof}
We can now define the notion of a guess being a good witness for a bounded until subformula and some position, according to the run $\rho$:
\begin{description} 
\item[Good witnesses for bounded Until.] From Lemma~\ref{lem:single-witness-in-interval}, there is at most one obligation $(\psi, x, t) \in \Theta_j$ such that $t \in I$. Let its identifier be $\psi \cdot k$, i.e., $\id_j(\psi, x, t) = \psi \cdot k$. Since $\rho$ is an accepting run, there is a least index $m > j$ s.t. $k \notin \range(\psi, \id_m)$. Define $\tau^\psi_j = \tau_m - \tau_j$. 
\end{description}

\begin{figure}
    \centering
    \begin{tikzpicture}
        \node (lgj) at (1,6.5) {\scriptsize $\Theta_j$:};
        \node (lgi1) at (1,5) {\scriptsize $\Theta_{i_1}$:};
        \node (lgik) at (1,2.5) {\scriptsize $\Theta_{i_k}$:};
        \node (lgm) at (1,1) {\scriptsize $\Theta_m$:};

        \node (gj) at (3.5,6.5) {$(\psi,x,t)$};
        \node (gi11) at (3.5,5) [gray] {$(\psi,x',t')$};
        \node (gi12) at (5.5,5) {$(\psi,x',t_{i_1})$};
        \node (gik1) at (3,2.5) [gray] {$(\psi,x'_{i_{k-1}},t'_{i_{k-1}})$};
        \node (gik2) at (5.5,2.5) {$(\psi,x'_{i_{k-1}},t_{i_k})$};
        \node (gm) at (5.5,1) {};
        \node (d) at (4.5,3.75) {$\dots$};

        \begin{scope}[->,dashed,>=stealth]
            \draw (gj) to (gi11);
            \draw (gik2) to (gm);
        \end{scope}
        \draw[<->,red] (7,6.5) to node [red,right] {$\tau^\psi_j$} (7,1);
    \end{tikzpicture}
    \caption{Tracking obligations with the same id in run $\rho$ of $\obred(\varphi)$}\label{fig:good-witness-in-interval}
\end{figure}
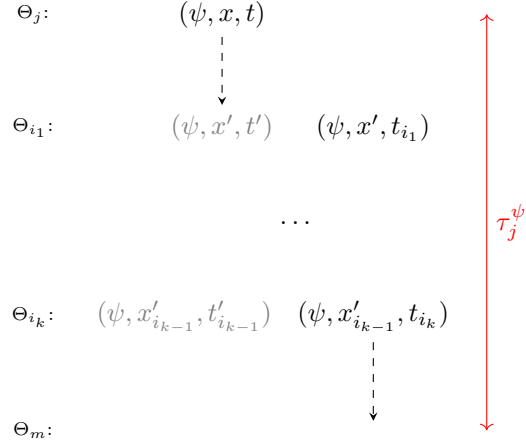
\begin{lemma}
    For some subformula $\psi$ and a position $j$ in the run $\rho$, $\tau^\psi_j \in I$.
\end{lemma}
\begin{proof}
    Firstly, if in the run $\rho$, $(\psi,x,t)$ was not removed due to any optimization rules, then $\tau^\psi_j$ will simply be $t$, and so $\tau^\psi_j \in I$. If not, let $i_{1},\dots,i_{k}$ be positions such that some optimization rule was applied to replace the current obligation with id $\psi \cdot k$. This is shown in Figure \ref{fig:good-witness-in-interval}, where the obligations removed are in gray. We will prove that at each such position $\ell$, if $(\psi,x'_{i_{\ell-1}},t'_{i_{\ell-1}})$ was replaced by $(\psi,x_{i_\ell},t_{i_\ell})$, then $\tau_{i_\ell} - \tau_j + t_{i_\ell} \in I$. This is enough because in the final position $i_k$, from the statement above, we have $\tau_{i_k} - \tau_j + t_{i_k} \in I$, and as $\rho$ is an accepting run, we know that $(\psi,x_{i_k},t_{i_k})$ will be eliminated for a position appearing after $t_{i_k}$ time elapse, meaning the defined $\tau^\psi_j \in I$ too. Now, to prove the required statement for a position $\ell$, we see that the only optimization rule that would have removed $(\psi,x'_{\ell-1},t'_{\ell-1})$ in the run would be \modifyold\, meaning a new obligation $(\psi,x_{\ell},t_{\ell})$ was generated with $t_{\ell} \in I$ and $x'_{\ell-1} + t_{\ell} \in I$, due to which $(\psi,x'_{\ell-1},t_{\ell})$ was added to $\Theta_{\ell}$. Because $x'_{\ell-1} = x + (\tau_\ell - \tau_j)$, $x'_{\ell-1} + t_\ell \in I$, and $t_\ell \in I$, we see that $(\tau_\ell - \tau_j) + t_\ell \in I$ too.
\end{proof}

Now, using $\rho$ and the observations above, we will construct $(\Theta'_i, \id'_i)$ for all $i \ge 1$ and claim that the following is an accepting run of $\ob(\varphi)$:
\begin{align*}
\rho':= \init_\varphi \xra{(\delta_1, a_1)} (\Theta'_1, \id'_1) \xra{\delta_2, a_2} (\Theta'_2, \id'_2) \cdots
\end{align*}
We will construct this run such that for every node $(\Theta'_i, \id'_i)$ and subformula $\psi$ where $\psi$ is a bounded until subformula, if $(\psi,0,t) \in \Theta'_i$, then $t = \tau^\psi_i$. 
In addition, for each $i \ge 1$, we will define a mapping $h_{i}: \Theta'_{i} \mapsto \Theta_{i}$ such that for every $\theta' \in \Theta'_i = (\psi',x',t')$, if $h_i(\theta') = (\psi,x,t)$, then $\psi=\psi'$. Intuitively, $h_i(\theta')$ gives the obligation $\theta$ in the reduced graph such that: if $\theta$ is satisfied, then $\theta'$ will also be satisfied. 
As we inductively build the run $\rho'$, the mapping $h_i$ will  help us define the required successor $(\Theta'_{i+1}, \id'_{i+1})$.   

\subparagraph*{Initial construction.} Firstly, we see that both $\rho$ and $\rho'$ will have $\init_{\varphi}$ as the initial node. Also, because initial obligation-sets for a subformula have at most one obligation of that subformula, and $\overline{\Theta}_{1}$ is well-formed, $\optimize(\overline{\Theta}_1,\overline{\id}_1) = (\overline{\Theta}_1,\overline{\id}_1)$. Thus, we can define $(\Theta'_1,\id'_1) = (\overline{\Theta}_1,\overline{\id}_1)$ itself, with $h_1$ being the identity function. 

Now, assuming we have constructed $(\Theta'_i, \id'_i)$ and we have the function $h_i$, satisfying the required properties, we will see how to define $(\Theta'_{i+1},\id'_{i+1})$ and $h_{i+1}$. 
\subparagraph*{Delay transition.} We need to first prove that the delay $\delta_{i+1}$ is possible from $(\Theta'_i,\id'_i)$. To do so, we see that for every bounded until $\psi$ type obligation present in $\Theta'_i$, at the time it was created, its waiting time was set as $\tau^\psi_j$ for some $j \leq i$, where $\tau^\psi_j = \delta_{j+1} + \dots + \delta_k$, $k > j$. This means that if such an obligation is present in $\Theta'_i$, its waiting time is either exactly $\delta_{i+1}$, or is greater than $\delta_{i+1}$. This means the node $(\Theta'_i,\id'_i) + \delta_{i+1}$ is defined.  

\subparagraph*{Constructing $(\Theta'_{i+1},\id'_{i+1})$.}
We will use the transition in Equation $\eqref{eq:original-transition}$ and $h_{i}$ to decide the transitions picked for the obligations of $\Theta'_{i} + \delta_{i+1}$. 
Looking at some $\theta' \in \Theta'_i$, where $h_{i}(\theta') = \theta$ for $\theta = (\psi,x,t)$, we will look at the $(\delta_{i+1},a_{i+1})$-extension of $\theta$ in transition $e_{i+1}\eqref{eq:original-transition}$. 
\begin{itemize}
    \item If $\theta+\delta_{i+1}$ is present in $\Theta''_{i+1}$, we will add $\theta'+\delta_{i+1}$ to $\Theta'_{i+1}$ and define $\id'_{i+1}(\theta'+\delta_{i+1}) = \id'_{i}(\theta')$.
    \item If a 'new' obligation of type $\psi'$, i.e. of the form $(\psi',0,t')$ is present in $\overline{\Theta}_{i+1}$:
    \begin{itemize}
        \item If $\psi'$ is a bounded until subformula, we will add $(\psi',0,t'')$ to $\Theta'_{i+1}$, where $t'' = \tau^{\psi'}_{i+1}$. We will set $\id'_{i+1}$ of this new obligation by picking the next available id for type $\psi'$.
        \item If $\psi'$ is not a bounded until subformula, we will simply add $(\psi',0,t')$ to $\Theta'_{i+1}$ and set its id to the next available one of type $\psi'$. 
    \end{itemize}
\end{itemize}
Notice that by our construction, as we add only one new obligation for each subformula type, we are getting a successor $\Theta'_{i+1}$ that is well-formed provided $\Theta'_i$ was well-formed. 
\subparagraph*{Defining $h_{i+1}$.} 
Given some $\theta' \in \Theta'_{i+1}$, we define $h_{i+1}(\theta')$ as follows:
\begin{itemize}
    \item If $\theta'$ is also present in $\Theta_{i+1}$, we just define $h_{i+1}(\theta') = \theta'$.
    \item If $\theta'$ is not present in $\Theta_{i+1}$, we look at the type of obligation $\theta'$, say $\psi'$:
    \begin{itemize}
        \item If $\psi'$ is a bounded until formula, we know from Lemma \ref{lem:single-witness-in-interval} that there will be a simple $\theta \in \Theta_{i+1}$ such that the waiting time of $\theta$ will be in the interval of $\psi'$. We will define $h_{i+1}(\theta') = \theta$ in this case.
        \item If $\psi'$ is an unbounded until subformula or a dual until subformula, we look at $\theta = (\psi',x,t) \in \Theta_{i+1}$ such that $x$ has the least value among other obligations of type $\psi'$ and define $h_{i+1}(\theta') = \theta$.
    \end{itemize}
\end{itemize}
Now, the constructed run $\rho'$ is an accepting run. This is because we define $h_i$ for each position $i$ which ensures that every obligation of type $\psi$ in $(\Theta'_i,\id'_i)$ is mapped to a formula of type $\psi$ in $(\Theta_i,\id_i)$. Also, the $id$ functions are defined so that if $\id'_i(\theta'_1) = \id'_{i+1}(\theta'_2)$ in $\rho'$ and $h_i(\theta'_1) = \theta_1$, $h_{i+1}(\theta'_2) = \theta_2$, then $\id_i(\theta_1) = \id_{i+1}(\theta_2)$. Finally, we know that $\rho$ is an accepting run, meaning every until obligation with a certain id has some node later in the run without that id. Hence $\rho'$ is also an accepting run. 

\begin{figure}
    \centering
    \scalebox{0.9}{
    \begin{tikzpicture}
    \begin{scope}[scale=0.8]
        \begin{scope}[rounded corners]
            \fill[red!30] (0.4,2.75) rectangle (1.6,2.25);
            \fill[blue!30] (1.2,1.25) rectangle (2.7,0.75);
        \end{scope}
        \draw[red,->] (-2,2.5) -- (-2,-3);
        \draw[red] (-2,2.5) to node[above] {\small $\tau^\varphi_1 = 2$} (-1.5,2.5);
        \draw[blue,->] (4,1) to (4,-5.5);
        \draw[blue,-] (3.5,1) -- (4,1); 
        \draw[blue,-] (3.5,-0.25) -- (4,-0.25); 
        \draw[blue,-] (3.5,-1.75) -- (4,-1.75); 
        \node[blue] (t1) at (5,1) {\small $\tau^{\varphi_1}_2 = 2.6$};
        \node[blue] (t1) at (5,-0.25) {\small $\tau^{\varphi_1}_3 = 1.4$};
        \node[blue] (t1) at (5,-1.75) {\small $\tau^{\varphi_1}_4 = 1.1$};
        
        \node (0) at (1, 3.5) {\scriptsize $\{\}$};
        \node (1) at (1, 2.5) {\scriptsize $\{(\varphi,0,2)\}$};        
        \node (2) at (1, 1) {\scriptsize $\{(\varphi,0.3,1.7),(\varphi_1,0,2.1)\}$};        
        \node (3) at (1, -0.25) {\scriptsize $\{ (\varphi, 1.5, 0.5), (\varphi_1, 1.2, 1.4)\}$}; 
        \node (4) at (1, -1.75) {\scriptsize $\{(\varphi, 1.8, 0.2), (\varphi_1, 1.5, 1.1)\}$};
        \node (5) at (1, -3) {\scriptsize $\{(\varphi_1, 1.7, 0.9)\}$};
        \node (6) at (1, -4.25) {\scriptsize $\{(\varphi_1, 2.1, 0.5)\}$};
        \node (7) at (1, -5.5) {\scriptsize $\{\}$};

        \node (l1) [purple] at (1,2.2) {\tiny $\varphi.1$};
        \node (l21) [purple] at (0,0.7) {\tiny $\varphi.1$}; 
        \node (l22) [purple] at (2,0.7) {\tiny $\varphi_1.1$};
        \node (l2) [] at (1,0.7) {};
        \node (l31) [purple] at (0,-0.55) {\tiny $\varphi.1$};
        \node (l32) [purple] at (2,-0.55) {\tiny $\varphi_1.1$};
        \node (l3) at (1,-0.55) {};
        \node (l41) [purple] at (0,-2.05) {\tiny $\varphi.1$};
        \node (l42) [purple] at (2,-2.05) {\tiny $\varphi_1.1$};
        \node (l4) [] at (1,-2.05) {};
        \node (l5) [purple] at (1,-3.3) {\tiny $\varphi_1.1$};
        \node (l6) [purple] at (1,-4.55) {\tiny $\varphi_1.2$};
    \end{scope}
    \begin{scope}[->, >=stealth]
        \draw (0) to node [right] {\tiny $0, b$} (1);
        \draw (l1) to node [right] {\tiny $0.3, a$} (2);
        \draw (l2) to node [right] {\tiny $1.2, a$} (3);
        \draw (l3) to node [right] {\tiny $0.3, b$} (4);
        \draw (l4) to node [right] {\tiny $0.2, b$} (5);
        \draw (l5) to node [right] {\tiny $0.4, a$} (6);
        \draw (l6) to node [right] {\tiny $0.5, a$} (7);
    \end{scope}
    \begin{scope}[xshift=9cm,scale=0.8]
        \node (0) at (0.5, 3.5) {\scriptsize $\{\}$};
        \node (1) at (0.5, 2.5) {\scriptsize $\{\textcolor{red!80}{(\varphi,0,2)}\}$};        
        \node (2) at (0.5, 1) {\scriptsize $\{(\varphi,0.3,1.7),\textcolor{blue}{(\varphi_1,0,2.6)}\}$};        
        \node (3) at (0.5, -0.25) {\scriptsize $\{ (\varphi, 1.5, 0.5), (\varphi_1, 1.2, 1.4), \textcolor{blue}{(\varphi_1, 0, 1.4)}\}$}; 
        \node (4) at (0.5, -1.75) {\scriptsize $\{(\varphi, 1.8, 0.2), (\varphi_1, 1.5, 1.1), (\varphi_1, 0.3, 1.1), \textcolor{blue}{(\varphi_1, 0, 1.1)} \}$};
        \node (5) at (0.5, -3) {\scriptsize $\{(\varphi_1, 1.7, 0.9), (\varphi_1, 0.5, 0.9), (\varphi_1, 0.2, 0.9) \}$};
        \node (6) at (0.5, -4.25) {\scriptsize $\{(\varphi_1, 0.9, 0.5), (\varphi_1, 0.6, 0.5)\}$};
        \node (7) at (0.5, -5.5) {\scriptsize $\{\}$};

        \node (l1) [purple] at (0.5,2.2) {\tiny $\varphi.1$};
        \node (l21) [purple] at (-0.5,0.7) {\tiny $\varphi.1$}; 
        \node (l22) [purple] at (1.5,0.7) {\tiny $\varphi_1.1$};
        \node (l2) [] at (0.5,0.7) {};
        \node (l31) [purple] at (-1.5,-0.55) {\tiny $\varphi.1$};
        \node (l32) [purple] at (0.5,-0.55) {\tiny $\varphi_1.1$};
        \node (l33) [purple] at (2.5,-0.55) {\tiny $\varphi_1.2$};
        \node (l41) [purple] at (-2.5,-2.05) {\tiny $\varphi.1$};
        \node (l42) [purple] at (-0.5,-2.05) {\tiny $\varphi_1.1$};
        \node (l43) [purple] at (1.5,-2.05) {\tiny $\varphi_1.2$};
        \node (l44) [purple] at (3.5,-2.05) {\tiny $\varphi_1.3$};
        \node (l4) [] at (0.5,-2.05) {};
        \node (l51) [purple] at (-1.6,-3.3) {\tiny $\varphi_1.1$};
        \node (l52) [purple] at (0.5,-3.3) {\tiny $\varphi_1.2$};
        \node (l53) [purple] at (2.5,-3.3) {\tiny $\varphi_1.3$};
        \node (l61) [purple] at (-0.5,-4.55) {\tiny $\varphi_1.2$};
        \node (l62) [purple] at (1.5,-4.55) {\tiny $\varphi_1.3$};
        \node (l6) [] at (0.5,-4.55) {};
    \end{scope}
    \begin{scope}[->, >=stealth]
        \draw (0) to node [right] {\tiny $0, b$} (1);
        \draw (l1) to node [right] {\tiny $0.3, a$} (2);
        \draw (l2) to node [right] {\tiny $1.2, a$} (3);
        \draw (l32) to node [right] {\tiny $0.3, b$} (4);
        \draw (l4) to node [right] {\tiny $0.2, b$} (5);
        \draw (l52) to node [right] {\tiny $0.4, a$} (6);
        \draw (l6) to node [right] {\tiny $0.5, a$} (7);
    \end{scope}
    \end{tikzpicture}
    }
    \caption{Getting a run in $\ob(\varphi)$ from a run in $\obred(\varphi)$}
    \label{fig:obred-until}
\end{figure}
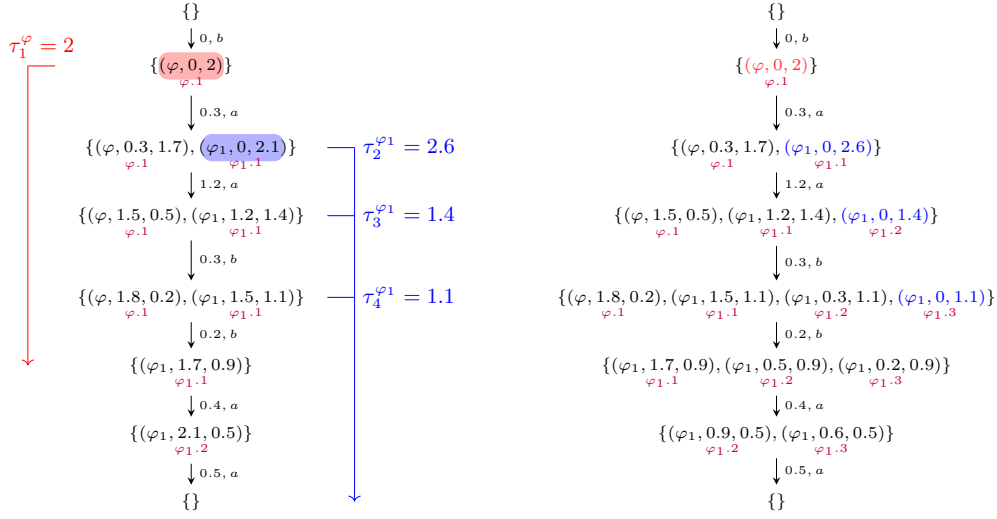

\begin{example}
    Figure \ref{fig:obred-until} illustrates the construction of $\rho'$ for a given $\rho$. Let us look at a couple of the constructed obligation-sets for $\rho'$ to get a better picture of the construction. We omit the details of the id functions as they are straightforward. 
    Firstly, for $\Theta_2 = \{(\varphi,0.3,1.7),(\varphi_1,0,2.1)\}$, the constructed $\Theta'_2$ contains obligation $(\varphi,0.3,1.7)$ because the obligation $(\varphi,0,2) + 0.3$ is present in $\overline{\Theta}_2$ (and in $\Theta_2$), and the newly generated obligation $(\varphi_1,0,2.6)$ because $\tau^{\varphi_1}_2 = 2.6$ (illustrated in the Figure in blue). Also, the mapping $h_1$ is defined as $h_2((\varphi,0.3,1.7)) = (\varphi,0.3,1.7)$ and $h_2((\varphi_1,0,2.6)) = (\varphi_1,0,2.1)$. Looking further, at $\Theta_3 = \{(\varphi,1.5,0.5),(\varphi_1,0,1.4)\}$ and $\Theta_4 = \{(\varphi,1.8,0.2),(\varphi_1,0.3,1.1)\}$, firstly we note that one possible transition $e_4$ is $\Theta_3 \xra{0.3,b} \{(\varphi,1.8,0.2),(\varphi_1,0.3,1.1),(\varphi_1,0,1.7)\}$. 
    Now, looking at the constructed $\rho'$ upto $\Theta'_3$, the constructed $h_3$ here is $h_3((\varphi,1.5,0.5)) = (\varphi,1.5,0.5)$ and $h_3((\varphi_1,1.2,1.4)) = h_3((\varphi_1,0,1.4)) = (\varphi_1,0,1.4)$. Now, to get $\Theta'_4$, we see that firstly $(\varphi,1.8,0.2)$ will get added to $\Theta'_4$ as before due to $h_3((\varphi,1.5,0.5)) + 0.3 \in \overline{\Theta}_4$. Similarly, $(\varphi_1,1.5,1.1)$ and $(\varphi_1,0.3,1.1)$ will be added to $\Theta_4$ too. Coming to the newly generated obligation, because $\tau^{\varphi_1}_4 = 1.1$, the obligation $(\varphi_1,0,1.1)$ will be added to $\Theta'_4$.
\end{example}

We now prove the other direction, that if $w$ has an accepting run in $\ob(\varphi)$, it has an accepting run in $\obred(\varphi)$. 

Consider an accepting run for $w$ in $\ob(\varphi)$:
\begin{align*}
\pi':= \init_\varphi \xra{(\delta_1, a_1)} (\Theta'_1, \id'_1) \xra{\delta_2, a_2} (\Theta'_2, \id'_2) \cdots
\end{align*}
Similar to earlier, we will use this run to construct an accepting run $\pi'$ for $w$ on $\obred(\varphi)$ of the form:
\begin{align*}
\pi:= \init_\varphi \xRightarrow{(\delta_1, a_1)} (\Theta_1, \id_1) \xRightarrow{\delta_2, a_2} (\Theta_2, \id_2) \cdots
\end{align*}
We will also require a mapping $g_i$ from $\Theta_i$ to $\Theta'_i$ to guide the successor computation. As mentioned before, $\optimize(\Theta'_1,\id'_1) = (\Theta'_1,\id'_1)$, and so the run $\init_\varphi \xRightarrow{\delta_1,a_1} (\Theta'_1,\id'_1)$ will remain the same in $\pi'$. To get $(\Theta_{i+1},\id_{i+1})$ and $g_{i+1}$ from $(\Theta_i,\id_i)$ and $g_i$, we will simply apply the optimize function and use $g_i$ to decide successors: for every obligation $\theta \in \optimize(\Theta'_i)$, the next transition for $\theta$ will be the transition for $g_i(\theta)$ in $(\Theta'_i,\id'_i) \xra{\delta_{i+1},a_{i+1}} (\Theta'_{i+1},\id'_{i+1})$, and so we will get $\Theta_{i+1}$ and accordingly get $\id_{i+1}$. 
We define $g_{i+1}$ as follows, where $\theta = (\psi,x,t)$ is some obligation in $\optimize(\Theta'_i,\id'_i)$:
\begin{itemize}
    \item[-] If $\psi$ is an unbounded until subformula, we know that the optimization rules either eliminate formulas, and $\theta$ will be present in $\Theta'_{i+1}$ too, so we will define $g_{i+1}(\theta)$ as that obligation.
    \item[-] If $\psi$ is an until subformula or a dual until subformula, we know from the optimization rules that either $\theta$ will be present in $\Theta'_{i+1}$, or there will at least one of the two obligations $\theta_1 = (\psi,x,t')$ or $\theta_2 = (\psi,x',t)$ for some $x',t'$. We will define $g_{i+1}(\theta)$ as $\theta$ if it is present, or $\theta_1$ if it is present, else as $\theta_2$.
\end{itemize}
We can now see that the $\pi'$ we construct is an accepting run, the proof being similar to our earlier reasoning: we define $g_i$ for each position $i$ which ensures that every obligation of type $\psi$ in $(\Theta_i,\id_i)$ is mapped to a formula of type $\psi$ in $(\Theta'_i,\id'_i)$.

\subsection{Proof of Theorem \ref{thm:red-bdd-mitl}}
\mitlbdd*
\begin{proof}
    We take an arbitrary node $(\Theta,\id)$ in $\obred(\varphi)$ such that there is a run from $\init_\varphi$ to $(\Theta,\id)$, say $\rho$, and show that there are at most $k_{\psi}$ obligations of type $\psi$ in $\Theta$. 
    \begin{itemize}
        \item If $\psi$ is a bounded until formula of the form $\psi_1 \U_I \psi_2$, we prove the statement using contradiction. We assume there are $K = k_{\psi}+1$ obligations of type $\psi$ in $\Theta$, say $\theta_1,\theta_2,\dots,\theta_{K}$ with $\theta_i = (\psi,x_i,t_i)$ for $1 \leq i \leq k_{\psi}+1$ and $x_1 > \dots > x_{K}$. Now, looking at two obligations $\theta_i, \theta_{i+1}$, as they are both present in $\Theta$, it means there was an obligation-set occuring before $\Theta$ in $\rho$ such that it contained obligations $\theta'_i=(\psi,x_i - x_{i+1},t_i + x_{i+1})$ and $\theta'_{i+1} = (\psi,0,t_{i+1}+x_{i+1})$, and the rule \keepboth\ was applicable on that position, meaning $t_i + x_{i+1} < I$ and $t_{i+1}+x_{i+1} + x_i - x_{i+1} = t_{i+1} + x_i > I$. As this will hold for each $\theta_i,\theta_{i+1}$ for $1 \leq i \leq K-1$, we see that for $1 \leq j \leq K-2$, $x_{j} - x_{j+2} > u-l$. Now, as $\varphi$ is an MITL formula, $u-l \geq 1$. Also, each $x_i \geq 0$ and $K \geq 2 + 2 \times \lceil \frac{l}{u-l} \rceil$, this means that $x_1 > u$. This is a contradiction because the waiting time for any obligation of type $\psi$ would be in $I$, and so $\theta_1$ cannot be active for a time longer than $u$.
        \item If $\psi$ is a bounded dual until formula of the form $\psi_1 \dual{\U}_I \psi_2$, we again prove the statement using contradiction. We assume there are $K = k_{\psi}+1$ obligations of type $\psi$ in $\Theta$, say $\theta_1,\theta_2,\dots,\theta_{K}$ with $\theta_i = (\psi,x_i,t_i)$ for $1 \leq i \leq K$ and $x_1 > \dots > x_{K}$. For some $\theta_i$ and $\theta_{i+1}$ in $\Theta$, as they are both present in $\Theta$, it means there was an obligation-set occuring before $\Theta$ in $\rho$ such that it contained obligations $\theta'_i=(\psi,x_i - x_{i+1},t_i + x_{i+1})$ and $\theta'_{i+1} = (\psi,0,t_{i+1}+x_{i+1})$. As the \keepboth\ rule was applicable here, $t_i + x_{i+1} < I$. Note that because the waiting time for a newly created $\psi$ obligation is always $u$, we have $t_i + x_i = u$. Using the two, we see that $t_i - t_{i+1} \geq u - l$. As this will hold for each $\theta_i,\theta_{i+1}$ for $1 \leq i \leq K-1$, and as previously stated, $u-l \geq 1$, this means that $t_1 > l$, which is a contradiction as the waiting time for any obligation of type $\psi$ would be exactly $u$, so $\theta_1$ cannot have a waiting time longer than $u$.
        \item If $\psi$ is an unbounded until or dual until formula or, by Definitions \ref{def:until_ubdd_opt} and \ref{def:dual_until_opt}, we see that at most 2 obligations or 1 obligation respevtively of type $\psi$ will be kept in any obligation-set, which is within the bound $K$, and so the statement holds in these case.
    \end{itemize}
    Hence, $\Theta$ has at most $k_{\psi}$ obligations of type $\psi$. 
\end{proof}

\section{Appendix for Section \ref{sec:symbolic}}
\begin{example}
    Looking at the example of the MITL formula $\varphi = \varphi_1 \U_{[2,10]} b$ for $\varphi_1 = \F_{[1,3]} a$ and with $M = 10$, looking at obligation-sets $\Theta = \{(\varphi,1.5,0.5),(\varphi_1,1.2,1.4)\}$ and $\Theta' = \{(\varphi,1.8,0.2),(\varphi_1,1.5,1.1)\}$, we see that $\Theta \regeq_M \Theta'$ because of the bijection $h$ s.t. $h((\varphi,1.5,0.5)) = (\varphi,1.8,0.2)$ and $h((\varphi_1,1.2,1.4)) = (\varphi_1,1.5,1.1)$.
\end{example}
\regionbisimulation*
\begin{proof}
Firstly, we show that $\regeq_{M}$ has a finite index. We see that we assign the first available id to newly created obligations, and there at at most $k_{\psi}$ obligations of type $\psi$ in any node. This along with the fact that the obligation-sets are well formed means that for any $(\Theta,id)$ in $\obred(\varphi)$ and any type $\psi$, $\range(\id,\psi) \subseteq \{1,\dots,k_{\psi}+1\}$, which gives us the finiteness of the number of equivalence classes of labelling functions. On the side of obligation-sets, finiteness of the number of equivalence classes simply follows from the fact that the timers are bounded and each equivalence class is given by: for each variable $y$ whether it is equal to $c$ or between $(c, c+1)$ for $0 \le c \le M$, and for all variables within the bound, an ordering of the fractional parts among them. 

Now, let $(\Theta, \id), (\Theta', \id')$ be two nodes such that $(\Theta,\id) \regeq_M (\Theta', \id')$. Consider a transition of $\obred(\varphi)$, $(\Theta, \id) \xRightarrow{\delta, a} (\Theta_2, \id_2)$. This can be broken down into three steps:
\begin{align*}
 (\Theta, \id) \xrightarrow{\delta} (\Theta + \delta, \id) \xrightarrow{a} (\Theta_1, \id_1) \xrightarrow{\optimize} (\Theta_2, \id_2)
\end{align*}
We will now show that each of these steps can be performed from $(\Theta', \id')$ with the analogous node being region equivalent. In other words we will construct the following sequence:
 \begin{align*}
 (\Theta', \id') \xrightarrow{\delta'} (\Theta' + \delta', \id') \xrightarrow{a} (\Theta'_1, \id'_1) \xrightarrow{\optimize} (\Theta'_2, \id'_2)
\end{align*}
such that  the corresponding intermediate nodes are region equivalent. 
\subparagraph*{Constructing $\delta'$.} Firstly, we see from the definition of region equivalence that we need a $\delta'$ s.t. $\integralclock{\delta'} = \integralclock{\delta}$, and a suitable value for the fractional value that we explain below. Let $Y$ be a set of values constructed as follows: for every $(\psi, x, t) \in \Theta$, if $x \le M$, then add $x$ to $Y$, if $t \neq \bot$, add $t$ to $Y$. Assume $Y = \{y_1, y_2, \dots, y_k\}$ and the ordering of fractional parts to be: $\fractional{y_1} \triangleleft_1 \fractional{y_2} \triangleleft_2 y_3 \cdots \triangleleft_{k-1} \fractional{y_k}$, with each $\triangleleft_j \in \{<, \le\}$. On a time elapse of $\delta$, each fractional value increases (even for timers, as we have defined) until they hit the next integer (when the fractional value becomes $0$ again) and then increase again. This results in a cyclic shift in the above ordering. Let the order of fractional values in $\Theta + \delta$ among the variables in $Y$ be $\fractional{y_i} \triangleleft_i \fractional{y_{i+1}} \triangleleft_{i+1} \cdots \fractional{y_k} < \fractional{y_1} \triangleleft_{1} \cdots \triangleleft_{i-1} \fractional{y_i}$. 

Now let us form a similar set of variables $Y'$ from $\Theta'$: for every $(\psi, x', t')\in \Theta'$, add $x'$ to $Y'$ if $\integralclock{x'} \le M$, and $t'$ to $Y'$ if $t' \neq \bot$. Since $\Theta' \regeq_M \Theta$, we have $Y' = \{ h(y) \mid y \in Y\}$ (where $h$ is the bijection witnessing the equivalence) and the ordering of fractional parts in $Y'$ is $\fractional{h(y_1)} \triangleleft_1 \fractional{h(y_2)} \triangleleft_2 \fractional{h(y_3)} \cdots \triangleleft_{k-1} \fractional{h(y_k)}$. We then choose such a $\delta'$ such that the ordering of fractional values in $\Theta' + \delta'$ is $\fractional{h(y_i)} \triangleleft_i \fractional{h(y_{i+1})} \triangleleft_{i+1} \cdots \fractional{h(y_k)} < \fractional{h(y_1)} \triangleleft_{1} \cdots \triangleleft_{i-1} \fractional{h(y_i)}$. This ensures $\Theta+\delta \regeq_M \Theta + \delta'$ by the same $h$. And because $(\Theta,\id) \regeq_M (\Theta',\id')$ by the bijection $h$, $(\Theta+\delta,\id) \regeq_{M} (\Theta'+\delta',\id')$ too.

\subparagraph*{Constructing $\Theta'_1$.} By definition, $\Theta_1$ is a minimal $a$-extension of $\Theta + \delta$, where we have written $a$-extension in short for $(0, a)$- extension. The node $\Theta_1$ retains some of the obligations from $\Theta + \delta$ and there are some fresh obligations. Let $R$ be the set of obligations in $\Theta + \delta$ that appear in $\Theta_1$ too. To construct $\Theta'_1$, we will retain the set $h(R) = \{ h(\theta) \mid \theta \in R\}$ and for the fresh obligations, we will choose values for the variables that enforce the region equivalence: for all clock variables, fresh obligations simply contain $0$; for timer variables appearing in Dual Until formulas, the value is given by the upper bound appearing in the interval and is therefore fixed without any choice; for timer variables appearing in Until formulas, we can choose a value that respects the integral values and the ordering of fractional parts, analogous to the way we chose $\delta'$. Computing $\id'_1$ the usual way, we see that because the ids for the obligations of $R$ will match because of the bijection $h$, and because we will add fresh obligations uniformly to $\Theta_1$ and $\Theta'_1$, we can see that $(\Theta_1,\id_1) \regeq_M (\Theta'_1,\id'_1)$.

\subparagraph*{Optimizing $\Theta'_1$ to get $\Theta'_2$.} Note that $\Theta_2$ is obtained from $\Theta_1$ by applying the optimization rules. Assume once again that we have an $h_1: \Theta_1 \to \Theta'_1$ witnessing the region equivalence. 

\begin{itemize}
\item Suppose $\theta$ is removed from $\Theta_1$ due to \removenew\ applied between $\theta', \theta \in \Theta$, then $h_1(\theta)$ can be removed from $\Theta'_1$ due to \removenew\ applied on $h_1(\theta'), h_1(\theta)$.
\item If \modifyold\ rule is applied to $\theta$ due to a fresh obligation $\theta'$, note that the same rule can be applied on $h_1(\theta)$ due to $h_1(\theta')$.
\item If \mmerge\ is applied between $\theta_1, \theta_2 \in \Theta_1$, then similarly, \mmerge\ can be applied between $h_1(\theta_1), h_1(\theta_2)$ in $\Theta'_1$ resulting in the associated variables being retained or removed. 
\item Finally, the rules applied on obligations with unbounded intervals can be analogously applied on $\Theta_1$ and $\Theta'_1$.
\end{itemize}
In summary, the $\optimize$ function retains or modifies the associated variables, with the resulting $\Theta'_2$ being a subset of the original values in $\Theta_2$. Since the original set of values was region equivalent, the retained subset in both $\Theta_2$ and $\Theta'_2$ continue to be region equivalent. After the function $\id'_2$ is computed the usual way, we see that we can use $h_1$ to get a corresponding bijection such that $(\Theta_2,\id_2) \regeq_M (\Theta'_2,\id'_2)$.
\end{proof}

\oblregcorr*
\begin{proof}
The proof of this fact is analagous to the correctness of the standard region equivalence for timed automata. Consider a run of $\obred(\varphi)$:
\begin{align*}
\rho: \init_\varphi \xRightarrow{\delta_0, a_0} (\Theta_0, \id_0) \xRightarrow{\delta_1, a_1} (\Theta_1, \id_1) \cdots
\end{align*}
By definition, there is a transition $\init_\varphi \xra{a_0} [\Theta_0, \id_0]_M$ and $[(\Theta_i, \id_i)]_M \xrightarrow{a} [\Theta_{i+1}, \id_{i+1}]$ for all $i \ge 1$ in the obligation-region graph. Morever, a node $(\Theta_i, \id_i)$ is accepting iff $[(\Theta_i, \id_i)]_M$ is accepting. Hence we get an infinite run of the obligation-region graph which is accepting iff $\rho$ is accepting.

For the other direction, assume we have a run in the obligation-region graph:
\begin{align*}
\rho': \init_\varphi \xRightarrow{a_0} [(\Theta_0, \id_0)]_M \xRightarrow{a_1} [(\Theta_1, \id_1)]_M \cdots
\end{align*}
For the initial transition, by definition we have $\init_\varphi \xra{\delta_0, a_0} (\Theta'_0, \id_0)$ in $\obred(\varphi)$ such that $(\Theta'_0, \id'_0) \regeq_M (\Theta_0, \id_0)$. Secondly, as $\regeq_M$ is a bisimulation, there is a transition $(\Theta'_0, \id'_0) \xra{\delta'_1, a_1} (\Theta'_1, \id'_1)$ such that $(\Theta'_1, \id'_1) \regeq_M (\Theta_1, \id_1)$. This way, we can continue to build a run, leading to the following run in $\obred(\varphi)$:
\begin{align*}
\rho: \init_\varphi \xra{\delta'_0, a_0} (\Theta'_0, \id'_0) \xra{\delta'_1, a_1} (\Theta'_1, \id'_1) \cdots
\end{align*}
Suppose $\rho$ is accepting. We need to show that $\rho'$ is accepting. Consider an Until obligation $\theta = (\psi, x, t)$ appearing in some $(\Theta_j, \id_j)$ (the obligation used to identify equivalence classes in $\rho'$). Now, consider $h(\theta)$ which appears in $(\Theta'_j, \id'_j)$. Since $\rho'$ is accepting there is an index $k > j$ where $\id_j(\theta)$ becomes inactive. By the bisimulation property, $h(\theta)$ has a sequence of transitions in $\rho$ until the index $k$ where it disappears. Hence $\id'_j(h(\theta))$ is inactive at $k$. This shows that $\rho'$ is accepting. A symmetric argument shows that if $\rho'$ is accepting, then $\rho$ is accepting as well. 
\end{proof}

\mitlcomp*
\begin{proof}
    Let $N$ be the number of until and dual until subformulas $\psi$ of some formula $\varphi$, $K$ be the maximum among the values $k_{\psi}$, and $M$ be the maximum value among the end points of intervals of $\varphi$ that are in $\N$. Then, looking at the region graph for $\varphi$, we see that the number of possible nodes is the number of equivalence classes of obligation-sets of $\obred(\varphi)$ with respect to $\regeq_M$. The possible number of equivalence classes for the naming function is in $\Oo(2^{K.N})$, because the ids for any $\psi$ are taken from the range $\{1,\dots,k_{\psi}+1\}$. The possible number of obligations in any node are at most $K\times N$, and so the possible number of variables present(clocks and timers) are at most $2 \times K \times N$. The equivalence classes of obligation-sets can be described by indicating whether each variable has value equal to an integral value $c$, or between $(c,c+1)$, or above $M$, meaning at most $2\times M + 2$ different values. Along with this, the equivalence classes also indicate an ordering of the fractional values of each variable, meaning a possible of $(N\times K)!$ different orderings. This means the number of equivalence classes for obligation-sets is in $\Oo(2^{N\cdot K}\times(N\cdot K)!\times (M)^{N \cdot K})$. As $(N \cdot K)!$ is bounded by $(N \cdot K)^(N \cdot K)$, we can rewrite $(N \cdot K)!$ by $2^{(N \cdot K) \log (N \cdot K)}$ inside the $\Oo$. Secondly $M^{(N \cdot K)}$ can be written as $2^{N \cdot K \cdot \log M}$. Therefore, we get the size to be bounded by $Oo(2^{N \cdot K \cdot (1 + \log N \cdot K + \log M)})$. When constants are encoded in binary, the term $K$ is exponential in the encoding: hence the term is at most doubly exponential whereas when the constants are in unary, the complexity is singly-exponential.

    For the MITL$_{0, \infty}$ fragment, $K = 2$, a constant which does not depend on the formula. This gives the PSPACE bound. 
\end{proof}
\end{document}